\documentclass[aps,pra,reprint,groupedaddress]{revtex4-1}
\usepackage[UKenglish]{babel}
\usepackage[]{graphicx, amsfonts, amsmath, amssymb, amstext, latexsym, float, color, hyperref, mathtools}
\newcommand{\be}{\begin{equation}}
\newcommand{\ee}{\end{equation}}

\newcommand{\eat}[1]{}

\newtheorem{theorem}{Theorem}

\newtheorem{corollary}[theorem]{Corollary}

\newtheorem{proposition}[theorem]{Proposition}
\newtheorem{remark}[theorem]{Remark}

\newenvironment{proof}[1][Proof]{\noindent\textbf{#1.} }{\ \hfill\rule{0.5em}{0.5em}}

\begin{document}

\title{Rate-loss analysis of an efficient quantum repeater architecture}

\author{Saikat Guha}
\thanks{Email of corresponding author: sguha@bbn.com}
\author{Hari Krovi}
\author{Christopher A. Fuchs$^\dagger$}
\author{Zachary Dutton}
\affiliation{Quantum Information Processing group, Raytheon BBN Technologies, 10 Moulton Street, Cambridge, MA USA 02138 \\
$^\dagger$Present address: Department of Physics, University of Massachusetts Boston, 100 Morrissey Blvd.,
Boston, MA USA 02125}
\author{Joshua A. Slater$^{\dagger\dagger}$}
\author{Christoph Simon}
\author{Wolfgang Tittel}
\affiliation{Institute for Quantum Science and Technology, and Dept. of Physics and Astronomy, University of Calgary, Alberta, T2N 1N4\\
$^{\dagger\dagger}$Present address: Vienna center for quantum science and technology (VCQ), Faculty of Physics, University of Vienna, 1090 Vienna, Austria}

%\date{\today}

\begin{abstract}
We analyze an entanglement-based quantum key distribution (QKD) architecture that uses a linear chain of quantum repeaters employing photon-pair sources, spectral-multiplexing, linear-optic Bell-state measurements, multi-mode quantum memories and classical-only error correction. Assuming perfect sources, we find an exact expression for the secret-key rate, and an analytical description of how errors propagate through the repeater chain, as a function of various loss and noise parameters of the devices. We show via an explicit analytical calculation, which separately addresses the effects of the principle non-idealities, that this scheme achieves a secret key rate that surpasses the TGW bound---a recently-found fundamental limit to the rate-vs.-loss scaling achievable by any QKD protocol over a direct optical link---thereby providing one of the first rigorous proofs of the efficacy of a repeater protocol. We explicitly calculate the end-to-end shared noisy quantum state generated by the repeater chain, which could be useful for analyzing the performance of other non-QKD quantum protocols that require establishing long-distance entanglement. We evaluate that shared state's fidelity and the achievable entanglement distillation rate, as a function of the number of repeater nodes, total range, and various loss and noise parameters of the system. We extend our theoretical analysis to encompass sources with non-zero two-pair-emission probability, using an efficient exact numerical evaluation of the quantum state propagation and measurements. We expect our results to spur formal rate-loss analysis of other repeater protocols, and also to provide useful abstractions to seed analyses of quantum networks of complex topologies.
\end{abstract}


\maketitle

Shared entanglement underlies many quantum information protocols such as quantum key distribution (QKD)~\cite{Eke91}, teleportation~\cite{Ben93} and dense coding~\cite{Ben92}, and is a fundamental information resource that can boost reliable classical and quantum communication rates over noisy quantum channels~\cite{Wil12, Wil12a}. Optical photons are arguably the only candidate for distributing entanglement across long distances. They however are susceptible to loss and noise in the channel, which is the bane of practical realizations of long-distance quantum communication. The maximum entanglement-generation rate over a lossy optical channel with no classical-communication assistance is zero when the total loss exceeds $3$ dB~\cite{Guh08}. With two-way classical-communication assistance, the rates achievable for entanglement generation, as well as those for reliable quantum communication and secret-key generation (i.e., QKD) over a lossy optical channel must decay linearly with the channel's transmittance (i.e., exponentially with optical fiber length), regardless of the specific protocol used, for loss exceeding $\sim 5$ dB~\cite{Tak13}, while the rate plunges to zero at a maximum loss threshold that is determined by the excess noise in the channel and detectors. In order to generate entanglement over long distances at high rates, intermediate nodes equipped with quantum processing power must be interspersed along the lossy channel. {\em Quantum repeaters} are one example of such nodes that can help circumvent the aforesaid linear rate-transmittance fall-off of the unassisted lossy channel---henceforth referred to as the Takeoka-Guha-Wilde (TGW) bound~\cite{Tak13}. However, {\em not} all quantum devices, for example quantum-limited phase-sensitive amplifiers, can serve as effective intermediate nodes for improved quantum communication performance over the unassisted pure-loss channel~\cite{Nam14}.

Several quantum repeater protocols have been proposed, most of which use entanglement swapping by Bell-state measurements, and quantum memories, of some form (see~\cite{San11} for a recent review). The basic quantum repeater protocol probabilistically connects a string of imperfect entangled qubit pairs by using a nested entanglement swapping and purification protocol, thereby creating a single distant pair of high fidelity~\cite{Bri98}. If used for QKD, those final distant entangled pairs are measured by Alice and Bob in randomly-chosen mutually-unbiased bases, followed by sifting, error-correction and privacy amplification over a two-way authenticated classical channel, to generate a shared secret. 

The original repeater protocol~\cite{Bri98} relied on purifying multiple long-distance imperfect shared entangled pairs (into fewer pairs of high fidelity)---a procedure known as {\em entanglement distillation}. As an alternative to entanglement distillation, several forward-quantum-error-corrected protocols have been proposed and analyzed~\cite{Jia09, Bra14}, which can afford a better rate performance at the expense of more frequent memory-based repeaters capable of universal quantum logic. Some of the more recently proposed forward-coded protocols do not even need any matter quantum memories, but come at the expense of requiring fast quantum logic and feedforward at all-optical center stations, as well as a potentially huge overhead in terms of the number of photons used for error correction~\cite{Mun12, Lo13}.

There is therefore a lot of interest in simpler approaches to quantum repeaters that do not use entanglement purification or quantum error correction. The seminal work in this area was the DLCZ protocol \cite{Dua01}, which uses single-photon interference to create entanglement between distant atomic ensemble quantum memories. This entanglement is swapped via linear optics and single-photon detections and finally converted into two-photon entanglement at the two endpoints using the same basic ingredients. The DLCZ protocol triggered a lot of experimental and theoretical activity~\cite{San11}. It has two key shortcomings from a practical point of view. First, the achievable entanglement distribution rate is very low. Second, its reliance on single-photon interference means that interferometric stability over long distances is required. A lot of subsequent work has focused on addressing these two points. One promising approach that addresses the first point is multiplexing. Refs.~\cite{Col07} and \cite{Sim07} proposed the use of spatial and temporal multiplexing respectively. The second point can be addressed by using two-photon interference instead of single-photon interference. Proposals based on two-photon interference include Refs.~\cite{Zha07, Yua08, Che07, San08}. The reader is also encouraged to see Ref.~\cite{San11} for a detailed review of Refs.~\cite{Col07, Sim07, Zha07, Che07, San08} and related work.

A more recent proposal by Ref. \cite{Sin13A} promises high entanglement distribution rates by combining two-photon interference and spectral multiplexing. It uses photon-pair sources, multi-mode quantum memories \cite{Sag11, PRA052329}, linear-optic Bell-state measurements~\cite{Bra95, Lut99}, and classical-only error correction. This protocol does not rely on purification, and does not require hierarchical connection of the elementary links (i.e., multiple connections can proceed simultaneously), and thus the memory coherence time requirements and the system's clock speed are not driven by long-distance classical communication delays. The protocol allows the fidelity (of the end-to-end shared entangled state) to deteriorate as the chain lengthens, and finally uses classical error correction on a long sifted sequence of correlated pairs of classical data generated by measurements by Alice and Bob, to extract quantum-secure shared secret keys.

Despite the practical appeal of the architecture proposed in~\cite{Sin13A}, a rigorous calculation of its achievable rate-vs.-loss performance---both entanglement-distillation and secret-key generation rates---in the presence of various loss and noise detriments, and showing that it can fundamentally outperform the TGW bound has yet to be done, and is the primary purpose of this paper. To our knowledge, we provide one of the first explicit calculations of the rate-vs.-loss function of any quantum repeater protocol, and hence a rigorous achievability proof that this repeater protocol can beat the TGW bound, even with lossy and noisy components. Our compact scaling results could help abstract off the rate-loss function of a linear repeater chain to seed future network theoretic analyses of quantum networks of more complex topologies. We hope that our work will incite similar detailed rate-loss analysis of other repeater protocols, which will enable quantitative resource-performance tradeoff-studies and comparisons of the various protocols.

A big challenge that faces practical designs of long-distance quantum repeater architectures is the quantitative understanding of how the shared entangled state evolves across concatenated swap operations across multiple repeater nodes, which would enable calculating the rates of various quantum communication protocols that may consume the generated shared entanglement. Some recent studies were done to analyze linear chains of quantum relays~\cite{Kha13} and memory-based repeaters~\cite{Raz08, Sin13A}, which have either used extensive numerical simulations, or proposed semi-analytic or approximate theoretical models. Another paper did an elaborate analysis of various prominent quantum repeater protocols from the perspective of evaluating the minimal required parameters to obtain a nonzero secret key at a given range~\cite{Sil13}. Finally, a recent study of a relay architecture constructed using spontaneous parametric downconversion (SPDC) sources and concatenated entanglement swapping~\cite{Kha15} suggests the need of quantum memories to beat the TGW bound.

In this paper, we present a complete analytical characterization of the evolution of the end-to-end shared-entangled state in a concatenated quantum repeater chain and evaluate its performance for QKD. We analyze the scheme proposed in~\cite{Sin13A}. We analyze QKD using the aforesaid repeater chain as an example application, and obtain an exact expression for the secret key rate as a function of loss, number of swap stages, and various loss-and-noise parameters of the channel and detectors. We account for fiber loss, detector dark counts, detector inefficiency, multi-pair emission rates of the entanglement sources, and loss in loading (readout) into (from) the quantum memories. We find a compact scaling law for how the quantum bit error rate (QBER)---the probability that Alice and Bob obtain a mismatched sifted key bit despite measuring their halves of the entangled state in the same bases---scales up with increasing number of swap levels. This analytical scaling has practical importance, since an experimentally measured QBER on a single elementary link can be used to predict the QBER (and hence the key rates) practically obtainable over a long-distance channel that is constructed with multiple elementary links made with identical imperfect devices. Our calculation involves a detailed analysis of the Bell-swap operations by modeling imperfect single-photon detectors with appropriate positive-operator-valued-measure (POVM) elements, and solving a variant of the {\em logistic map}, a non-linear difference equation whose solutions are known to be chaotic in general~\cite{Sch1870}. Our calculations show that the aforesaid repeater chain, even if built using lossy and noisy devices, attains an overall rate-loss scaling for QKD that outperforms the TGW bound---the best performance achievable by any QKD protocol that does not employ quantum repeaters. To be precise, if $\eta \in (0, 1]$ is the end-to-end transmittance of the Alice-to-Bob channel, we show that by dividing up the channel into an optimum number of repeater nodes, the secret key rate achieved by the repeater chain, $R = A\eta^\xi$. The pre-factor $A$ and the power-law exponent $\xi$, $0 < \xi < 1$ are constants that are functions of various loss and noise parameters of the system. This beats the TGW bound's rate-loss scaling, i.e., $R \le \log[(1+\eta)/(1-\eta)] \approx 2.89\eta$ bits/mode, for $\eta \ll 1$~\cite{Tak13}. Furthermore, since we calculate the exact quantum state after every swap stage, our results can be used to calculate any other quantity of interest, such as fidelity (see Appendix~\ref{app:QBER_state}), for other applications of long-distance shared entanglement.

We also do an exact evaluation of the repeater chain numerically---using an efficient routine that employs sparsified matrix representations of bosonic operations---which enables us to go beyond sources with zero two-pair emissions, i.e., $p(2) > 0$. Even for sources with $p(2)>0$, our analytical prediction of QBER propagation through the repeater chain is shown to hold, albeit with a $p(2)$-dependent modification to a pre-factor. Using the above phenomenological model of QBER propagation, we show that positive two-pair probability $p(2)$ is shown to deteriorate the rate-distance function, but in the following way---at any given value of $p(2)$, there is a maximum number $N_{\rm max}(p(2)) \approx 1 + c/p(2)$ of elementary links such that for $N < N_{\rm max}$ links, the rate-loss envelope achieved by the repeater chain remains almost identical to what is achieved by a $p(2)=0$ source ($c$ is a constant), and thus continues to beat the TGW bound's scaling limit. However, for a chain with $N$ links with $N \ge N_{\rm max}(p(2))$, the key rate becomes worse at all range $L$ compared to when fewer than $N$ elementary links are employed. Conversely for a given $N$, as long as $p(2)$ is less than the inverse of the function $N_{\rm max}(p(2))$, the rate-loss envelope remains practically unaffected.

The paper is organized as follows. We begin with a description of the repeater architecture, and set notations, in Section~\ref{sec:Architecture}. In Section~\ref{sec:analysis}, we state our main results, followed by a high-level description of the key steps of our theoretical analysis. All the detailed proofs are deferred to the Appendices. We then summarize our main numerical results in Section~\ref{sec:numerics}, and an empirical analysis of the effect of source imperfections on the scaling of the secret key rate. Finally, we conclude the paper in Section~\ref{sec:conclusions}, with thoughts on open questions and future work.

\section{The repeater architecture}\label{sec:Architecture}

The architecture~\cite{Sin13A} is depicted schematically in Figs.~\ref{fig:architecture},~\ref{fig:connections}, and~\ref{fig:timing_diagram}. The total Alice to Bob range, $L$ km of lossy fiber, is divided into $N = 2^n$ elementary links.
\begin{figure}
\centering
\includegraphics[width=\columnwidth]{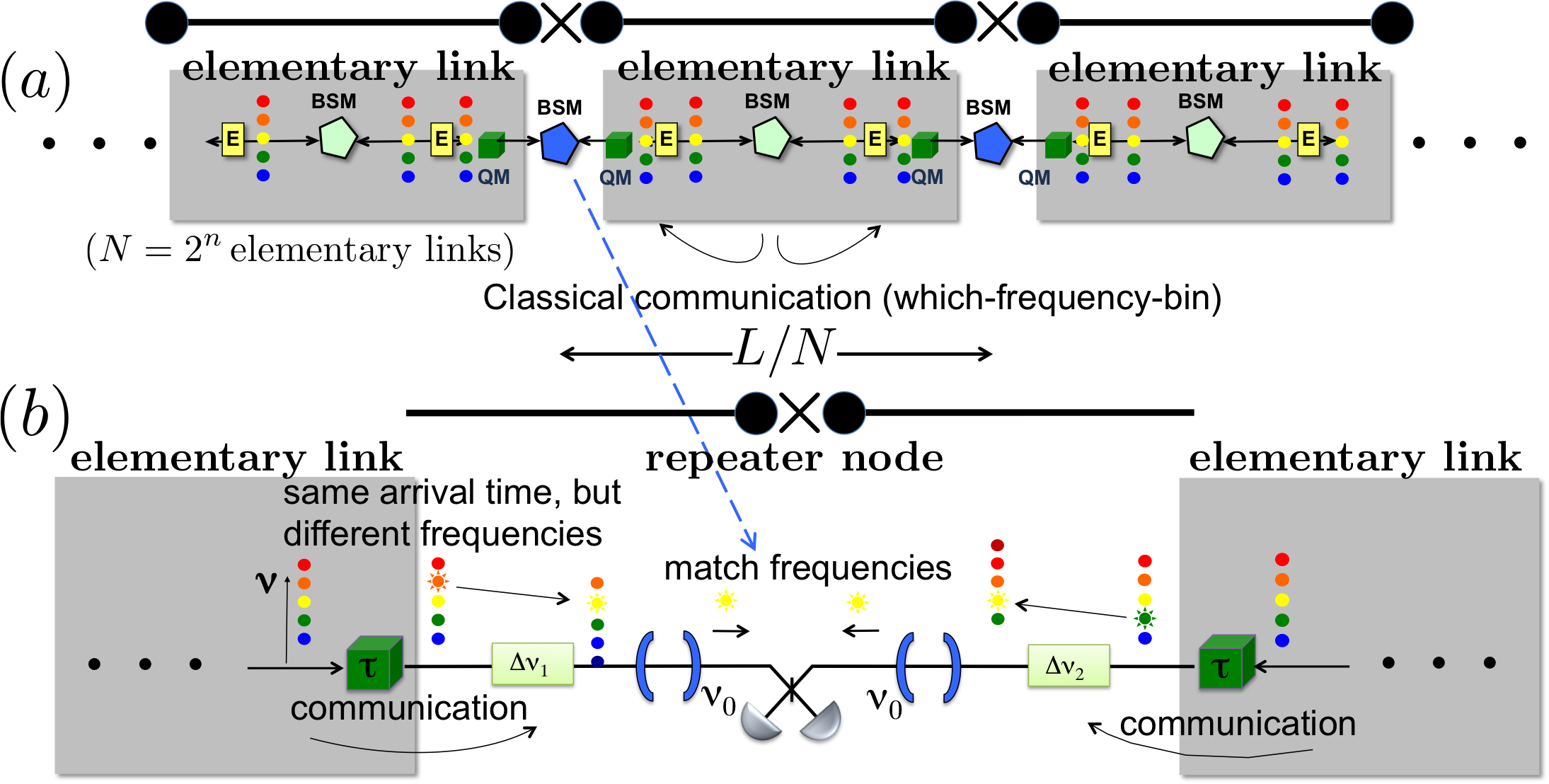}
\caption{(Color online) Schematic of quantum repeater architecture~\cite{Sin13A}.}
\label{fig:architecture}
\end{figure}

{\em The elementary links.}---Entangled photon-pair sources (E) at the two ends of each elementary link produce an $M$-fold tensor product maximally-entangled Bell state, i.e., $|M^{\pm}\rangle^{\otimes M}$, $|M^{\pm}\rangle \triangleq \left[|10,01\rangle \pm |01,10\rangle\right]/\sqrt{2}$, where $M$ is the number of orthogonal frequency modes. The sources then send halves of this entangled state towards the link's center. The other halves are loaded to multi-mode atomic quantum memories (QM) at each end of the elementary link~\cite{Sag11, PRA052329} (see Fig.~\ref{fig:architecture}). Each qubit of the Bell pair is encoded in two time-resolved bins ($\left\{|10\rangle, |01\rangle\right\}$). Each qubit (over all $M$ orthogonal frequency modes) occupies $T_q$ seconds, and undergoes lossy transmission with transmittance $\lambda = 10^{-(\alpha L / 2N)/10}$, where $\alpha$ (in dB/km) is the fiber's loss coefficient. At the center of the link, linear-optic Bell-state measurements (BSMs)~\cite{Bra95} act on the $M$ qubit pairs. The BSM comprises a 50-50 beam-splitter followed by a pair of single-photon detectors (which acts in sequence on each of the two time bins of the qubit) that can spectrally-resolve $M$ frequency modes. We assume however that the detectors have no photon number resolution. The detection efficiency and dark-click probability (per frequency mode and time bin) for each detector is taken to be $\eta_e$ and $P_e$, respectively. A linear-optic BSM is successful with at most $50\%$ probability~\cite{Lut99}. The sources $E$ are assumed to be deterministic~\cite{Dou10, Hua12}, i.e., they generate a copy of $|M^\pm\rangle^{\otimes M}$, every $T_q$ seconds, over the $M$ orthogonal frequency modes. This suffices since any zero-photon emission probability can be subsumed into the detection efficiency $\eta_e$, as we will see later. Non-zero two-pair emission probability $p(2)$ will be addressed in Section~\ref{sec:numerics}. Upon successful projection by the BSM on one of the Bell states in at least one of the $M$ frequencies, which happens with probability $P_s(1) = 1 - (1 - P_{s0})^M$, the BSM communicates the which-frequency-was-successful information to both ends. $P_{s0}$ is the success probability for a single frequency. We denote the (two-qubit four-mode) quantum state of a successfully-created elementary link, $\rho_1$.

{\em Connecting elementary links.}---The two memories at a {\em repeater node}, upon receipt of a pair of which-frequency information from the adjoining elementary links, translate their qubits to one pre-determined common frequency. A BSM at a {\em single} frequency is then performed on this pair~\cite{Sin13A}. The BSMs at the elementary-link centers all proceed simultaneously, and so do the repeater-node BSMs. This is unlike the DLCZ protocol, where BSMs are performed hierarchically, necessitating longer-lifetime memories. We assume a universal synchronized clock is available. The clock-rate of the system ($T_q^{-1}$) is limited by the time it takes to perform the BSMs at the elementary link centers ($\tau_{\rm BSM}$), those at the repeaters ($\tau_{\rm BSM}^\prime$), and the time for loading (readout) of the qubits to (from) the memories, $\tau_{\rm mem}$. There is a latency between entangled pair emissions and secret key generation, but the clock rate is not tied to this latency (see Fig.~\ref{fig:timing_diagram} for the timing diagram). We denote the efficiencies and dark-click probability for each detector used for the repeater-node BSMs, $\eta_r$ and $P_r$, respectively. Let $\lambda_m$ denote the sub-unity efficiency in loading (and retrieving) the photonic qubit into (and from) the memories, and that of frequency shifting and filtering. If this BSM is successful, two elementary links are connected to form a two-qubit entangled state $\rho_2$. Two copies of $\rho_2$ are connected (probabilistically) to produce $\rho_3$, etc. (although, as noted above, the connections do not have to proceed in this hierarchical manner). Given two identical successfully-heralded copies of $\rho_{i-1}$, the probability that a repeater-node BSM successfully heralds a $\rho_i$, is $P_s(i)$, and as we will see later, $P_s(i) = P_s$, $\forall i \in \left\{2, \ldots, n+1\right\}$.

{\em Error probabilities and key rate.}---Say, Alice and Bob make measurements on the two-qubit shared state $\rho_{i}$, either in the computational basis (single-photon detection on each of the two modes of their respective qubits), or the $45$-degrees rotated basis (realized by a 50-50 beamsplitter action on the two modes of their respective qubits, followed by single-photon detection on each mode). The detection efficiency and dark-click probability of their detectors are denoted $\eta_d$ and $P_d$. Alice and Bob then share their detection outcomes over an authenticated public channel. This detection of one copy of $\rho_i$ produces one of $16$ possible outcomes. As an example, the detection outcome ``$1,0;1,1$" means Alice gets a click and a no-click outcome on her qubit, and Bob gets clicks on detection of both modes of his qubit (it is instructive to note here that the ``$1,1$" outcome is possible only if $P_d > 0$). The sift probability $P_1$ is the probability that neither Alice nor Bob get zero clicks on both their detectors (i.e., $9$ of the $16$ possible outcomes), {\em given} they both measure their qubits in the same basis~\footnote{Note that this definition of sifting clearly suggests that, if the entanglement sources have a non-zero two-pair-emission probability $p(2)$, then an improved sifting performance could be obtained if Alice's and Bob's detectors have photon number resolving (PNR) capability, since that will help post-select out erroneous multi-photon events. We explore and analyze this further in~\cite{Kro15}.}. Upon a successful sift, Alice interprets her sifted bit as: ``$0,1$" $\to 0$, ``$1,0$" $\to 1$, and ``$1,1$" $\to 0$ or $1$ with equal probability, whereas Bob interprets his sifted bit as: ``$0,1$" $\to 1$, ``$1,0$" $\to 0$, and ``$1,1$" $\to 0$ or $1$ with equal probability. One may wonder why Alice and Bob do not simply discard all the two-click events as errors (in which case the sift would happen conditioned only on $4$ of the $16$ possible measurement outcomes). Doing so exposes them to a security vulnerability that was identified by L{\" u}tkenhaus in~\cite{Lut99a}. Conditioned on a successful sift, we denote $Q_i$, the QBER, to be the probability that the sifted bits Alice and Bob infer are different. The error correcting code used to extract keys must code around this error rate. If all detectors are noiseless (i.e., $P_e = P_r = P_d = 0$), $Q_i = 0$, $1 \le i \le n+1$. The overall success probability in creating the shared state $\rho_{n+1}$,
$
P_{\rm succ} = P_s(n+1) \,\left(P_s(n)\right)^2 \ldots \left(P_s(2)\right)^{2^{n-1}}\, \left(P_s(1)\right)^{2^n} = P_s^{N-1}P_s(1)^N
$, $N = 2^n$.
Let us assume Alice and Bob make the aforesaid measurement and sifting on $K$ identical copies of the qubit-pair $\rho_{n+1}$, i.e., a shared state created by connecting $N=2^n$ elementary links. In the limit of large $K$, and assuming an optimal error correcting code, Alice and Bob can extract ${P_1P_{\rm succ}R_2(Q_{n+1})}/2$ unconditionally-secure secret key bits per qubit pair. Therefore, the secret-key rate is given by,
\begin{equation}
R = {P_1P_{\rm succ}R_2(Q_{n+1})}/{2T_q}\;{\text{secret-key bits/s}},
\label{eq:keyrate}
\end{equation}
where the factor of $2$ in the denominator accounts for the probability that Alice and Bob use the same basis choice, $R_2(Q) = 1 + 2(1 - Q)\log_2(1 - Q) + 2Q\log_2(Q)$ is the secret-key rate of BB84 in bits per sifted symbol~\cite{ShorPreskill}, with $Q$ the error probability in the sifted bit. Fig.~\ref{fig:timing_diagram} shows a pictorial description of the entire process described in this section. Refs.~\cite{Fer12, Xio12} generalized~\eqref{eq:keyrate} for the case when Alice and Bob use a $d$-dimensional encoding ($d > 2$), and $g$ mutually-unbiased measurement bases, $2 \le g \le d+1$.

\begin{figure}
\centering
\includegraphics[width=\columnwidth]{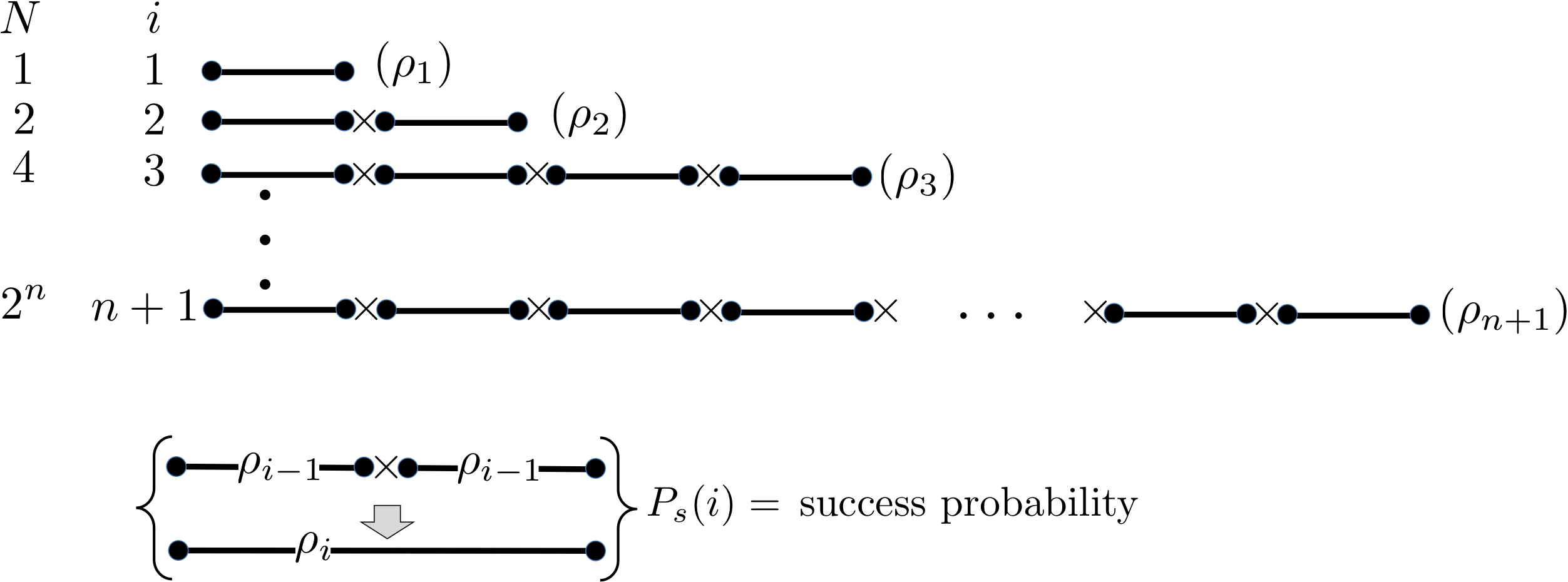}
\caption{Concatenated linking of $N=2^n$ elementary links. Each black dot is one qubit, comprising two temporal modes at one standard center frequency.}
\label{fig:connections}
\end{figure}

\begin{figure*}
\centering
\includegraphics[width=\textwidth]{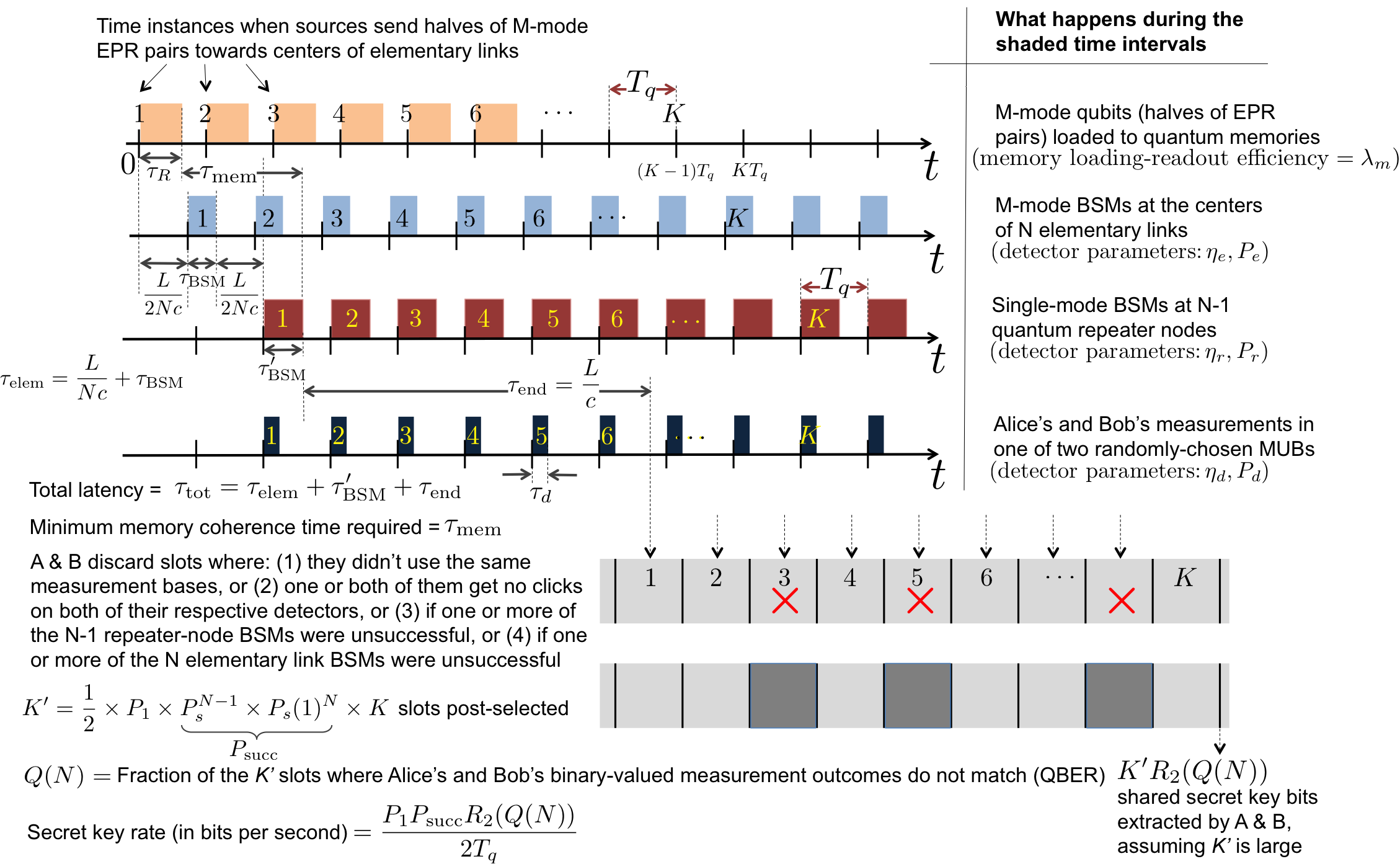}
\caption{(Color online) Timing diagram for the operation of the repeater architecture. At times $t = kT_q$, $k = 0, 1, \ldots$, the sources synchronously generate and send $M$-mode EPR halves towards centers of elementary links (which ideally take time $L/2Nc$ to arrive at the links' centers assuming $c$ to be the speed of light in fiber), while they load the other entangled halves into local quantum memories. The elementary link BSMs takes time $\tau_{\rm BSM}$, and the (classical) which-frequency-succeeded information takes time $L/2Nc$ to arrive back at the repeater nodes. At this point, each repeater node (synchronously) attempts a local BSM at a common frequency across the two qubits held in the two memories linked to the elementary links on its two sides, which takes time $\tau_{\rm BSM}^\prime$. The one-bit classical results of these BSMs take up to $\tau_{\rm end} = L/c$ seconds to reach Alice and Bob. Synchronously with the repeater-node BSMs, Alice and Bob measure the qubits in their respective quantum memories, which takes time $\tau_d \le \tau_{\rm BSM}^\prime$, we assume. Once the one-bit (success or failure) outcomes from all the repeater nodes arrive at Alice and Bob, they begin their classical processing. This involves first discarding the instances when they did not use matching measurement bases, those when they did use the same bases but did not get a successful sift event, and those when not all repeater nodes succeeded. Thereafter they use error correction to sieve out shared secret bits.}
\label{fig:timing_diagram}
\end{figure*}

\section{Theoretical analysis of the quantum repeater chain}\label{sec:analysis}

In Section~\ref{sec:state_propagation}, we will summarize our results on the full analytical characterization of the end-to-end shared entangled state $\rho_i$, $1 \le i \le n+1$, generated by the repeater chain (which could be useful in analyzing other non-QKD applications as well). We summarize explicit formulas for $P_{\rm succ}$, $P_1$, and $Q_{n+1}$, using which we calculate the secret key rate using Eq.~\eqref{eq:keyrate}. In Section~\ref{sec:ratelossenvelope}, we show that the key rate $R_N(L)$ vs. the Alice-to-Bob range $L$ when $N$ equal-length elementary links are employed, is described approximately by a three-segment plot. Using this characterization of $R_N(L)$, we derive the rate-vs.-distance envelope $R(L)$ attained by the repeater chain when an optimal number of elementary links is employed for any given total range, and show that the rate achieved by the repeater protocol is given by $R(L) = A\eta^\xi$, where $\eta = e^{-\alpha L}$, and $\xi < 1$, hence proving that it beats the TGW bound's scaling limit---the best rate-loss scaling achievable without assistance of quantum repeaters (which translates to, $\xi = 1$). Throughout Section~\ref{sec:analysis}, we provide proof sketches, deferring all detailed calculations to the Appendices.

\subsection{Shared state propagation and secret-key rate}\label{sec:state_propagation}

\begin{theorem}\label{thm:Q_main}
Assuming Alice and Bob make a measurement on $\rho_i$ in the same basis,
\begin{enumerate}
\item {\em Sift probability.} The probability Alice and Bob use the same measurement bases is $1/2$. Conditioned on them using the same bases, the probability of a successful sift (i.e., them deeming their measurement outcomes usable for further processing) is given by, $P_1 = (q_1 + q_2 + q_3)^2$, where $q_1 = (1-P_d)A_d$, $q_2 = (1-A_d)P_d$, $q_3 = P_dA_d$, $A_d \equiv \eta_d + (1 - \eta_d)P_d$, are defined in terms of loss and noise parameters of Alice's and Bob's detectors.

\item {\em QBER.} Conditioned on a successful sift, the error probability $Q_i$, i.e., the probability that Alice and Bob obtain mismatched bits, is given by,
\begin{equation}\label{eq:Q_exactform}
Q_i = \frac{1}{2}\left[1 - \frac{t_d}{t_r}\left(t_rt_e\right)^{2^{i-1}}\right], \, 1 \le i \le n+1,
\end{equation}
where $t_e = (1-2w_1)/(1+2w_1)$, $t_r = (1-2w_r)/(1+2w_r)$, and $t_d = ((q_1-q_2)/(q_1+q_2+q_3))^2$ are functions of loss-noise parameters of detectors in the elementary links, memory (repeater) nodes, and Alice-Bob, respectively. The parameters $t_x$ become one when the respective detectors ($x = e, d, r$) have zero dark-click probability, i.e., $P_x = 0$ (but may have sub-unity detection efficiency, i.e., $\eta_x < 1$). $2w_1 = 2c_e/(a_e+b_e)$, is the relative probability of classical correlations to that of pure Bell states in the elementary link state, $\rho_1$. $2w_r = 2c/(a+b)$ is the fractional probability spillovers to the classically-correlated states at each repeater connection. See Proposition~\ref{prop:connections} for definitions of $a_e, b_e, c_e, a, b$ and $c$ in terms of various loss and noise parameters.

\item {\em Successful connection probabilities.} The success probability $P_s(i)$, to prepare $\rho_i$ from two copies of $\rho_{i-1}$, is given by: $P_s(1) = s_1 = a_e + b_e + 2c_e$, and $P_s(i) = s = a + b + 2c$, for $2 \le i \le n+1$. The overall success probability, $P_{\rm succ} = P_s^{N-1}P_s(1)^N$;
\begin{equation}
P_{\rm succ} = \frac{1}{4s}\left[4s\left(1-(1-4s_1)^M\right)\right]^{2^n}.
\end{equation}
\end{enumerate}
\end{theorem}

\begin{proof}
(sketch)---The proof of Theorem~\ref{thm:Q_main} involves a detailed analysis of how the quantum states $\rho_i$ evolve through successive connections of elementary links (sketched in Fig.~\ref{fig:connections}) and finding the exact solution of a variation of the so called {\em logistic map}, whose solutions are chaotic in general. With the $Q_i$ as defined above, it is easy to see that the following recursive relation holds:
\begin{equation}\label{eq:errorpropagation}
(1 - 2Q_{i+1}) = \frac{t_r}{t_d}(1-2Q_i)^2, 1 \le i \le n.
\end{equation}
The pre-factor $t_r/t_d$ in the above error-propagation law equals one if the detectors at the memory nodes have zero dark clicks ($P_r = 0 \Rightarrow t_r = 1$) {\em and} if the detectors used to measure the end points of $\rho_i$ have zero dark clicks ($P_d = 0 \Rightarrow t_d = 1$). The constant $t_r$ is only a function of the fractional probability transferred to classical correlations ($2c$) to that which goes to one of two Bell states ($a+b$), when two pure Bell states are connected by a linear-optic BSM with lossy-noisy detectors (see Proposition~\ref{prop:connections}). We note that the constant $t_r/t_d$ does not depend upon the parameters that specify the quality of the elementary link, but $Q_1$, the QBER of the elementary link, does depend upon the elementary-link parameters.

We now describe the steps leading up to the proof of the expressions in Theorem~\ref{thm:Q_main}. We will defer several details to Appendices~\ref{app:elem}, \ref{app:success}, \ref{app:sift}, \ref{app:QBER}, and \ref{app:logisticmap}. We assume without loss of generality that the sources always produce the state $|M^+\rangle^{\otimes M}$. In reality, the sources may produce $|M^+\rangle$ or $|M^-\rangle$ in each mode probabilistically, but if the signs are known a posteriori (as in an SPDC source), they can be accounted for in post processing at the error-correction stage. In fact, as long as the sources produce any one of the four Bell-basis states in each $T_q$ second, if it is known which one was produced, it can be accounted for in classical post-processing. Let us first consider calculating $\rho_i$, the two-qubit state after successfully connecting $2^{i-1}$ elementary links.

\begin{proposition}\label{prop:connections}
The quantum state $\rho_i$ obtained after $i$ connection levels, $1 \le i \le n+1$, is given by,
\begin{eqnarray}
\rho_i &=& \frac{1}{s_i}\left[a_i|M^+\rangle\langle M^+| + b_i|M^-\rangle\langle M^-| + c_i|\psi_0\rangle\langle\psi_0|\right. \nonumber \\
&&\left. + d_i|\psi_1\rangle\langle\psi_1| + d_i|\psi_2\rangle\langle\psi_2| + c_i|\psi_3\rangle\langle\psi_3|\right],
\end{eqnarray}
where $|\psi_0\rangle = |01,01\rangle$, $|\psi_1\rangle = |01,10\rangle$, $|\psi_2\rangle = |10,01\rangle$, $|\psi_3\rangle = |10,10\rangle$, $|M^{\pm}\rangle = \left[|\psi_2\rangle \pm |\psi_1\rangle\right]/\sqrt{2}$, $s_i = a_i + b_i + 2(c_i + d_i)$ is a normalization constant, and the coefficients of the state $\rho_{i+1}$ are recursively given as:
\begin{eqnarray}
a_{i+1} &=& \frac{1}{s_i^2}\left[aa_i^2 + (a+b)a_ib_i + bb_i^2\right], \label{eq:aip1}\\
b_{i+1} &=& \frac{1}{s_i^2}\left[ba_i^2 + (a+b)a_ib_i + ab_i^2\right], \label{eq:bip1}\\
c_{i+1} &=& \frac{1}{s_i^2}\left[c(a_i + b_i)^2 + 2(a+b)c_i(a_i + b_i + 2d_i) \right. \nonumber\\
&&\left. + 4c(d_i(a_i + b_i) + c_i^2 + d_i^2)\right], \label{eq:cip1}\\
d_{i+1} &=& \frac{1}{s_i^2}\left[4cc_i(a_i + b_i + 2d_i) + 2(a+b)(d_i(a_i + b_i)\right. \nonumber\\
&&\left. + c_i^2 + d_i^2)\right], \,{\text{with}}\label{eq:dip1}\\
s_{i+1} &=& a_{i+1} + b_{i+1} + 2(c_{i+1} + d_{i+1}),\label{eq:sip1}
\end{eqnarray}
where the parameters,
\begin{eqnarray}
a &=& \frac{1}{8}\left[P_r^2(1-A_r)^2 + A_r^2(1-P_r)^2\right], \label{eq:a}\\
b &=& \frac{1}{8}\left[2A_rP_r(1-A_r)(1-P_r)\right], \label{eq:b}\\
c &=& \frac{1}{8}P_r(1-P_r)\left[P_r(1-B_r)+B_r(1-P_r)\right],\label{eq:c}
\end{eqnarray}
with $A_r = \eta_r\lambda_m + P_r(1 - \eta_r\lambda_m)$, and $B_r = 1 - (1-P_r)(1 - \eta_r\lambda_m)^2$, are functions of the system's loss and noise parameters. For $i=1$ (the elementary link), we have the initial conditions, $a_1 = a_e$, $b_1 = b_e$, $c_1 = c_e$, and $d_1 = 0$, with $s_1=a_e+b_e+2c_e$, where $a_e, b_e, c_e$ are defined exactly as $a, b, c$, with $(P_e, A_e, B_e)$ replacing $(P_r, A_r, B_r)$ in Eqs.~\eqref{eq:aip1},~\eqref{eq:bip1},~\eqref{eq:cip1}, where $A_e = \eta_e\lambda + P_e(1 - \eta_e\lambda)$ and $B_e = 1-(1-P_e)(1 - \eta_e\lambda)^2$, defined similar to $A_r$, $B_r$. Here, $\lambda_m$ is the efficiency of loading (reading) the photonic qubits into (from) the memories, and $\lambda = e^{-\alpha L/2N}$ is the channel transmittance of half of an elementary link.
\end{proposition}

\begin{proof}
(sketch) A detailed proof is given in Appendix~\ref{app:elem}, where we calculate the state $\rho_i$ (i.e., the coefficients $a_i, b_i, c_i, d_i$) explicitly for all $i$ explicitly in terms of the loss and noise parameters. The key steps are: (i) to realize that $\lambda$ and $\lambda_m$ can be subsumed in the detector efficiencies $\eta_e$ and $\eta_r$ of the BSMs, respectively, thereby rendering all qubit transmissions lossless, (ii) realizing that a single-photon detector of efficiency $\eta$ and dark-click probability $P$---when the impinging light is guaranteed to have no more than $2$ photons---is accurately described by the POVM elements (see Fig.~\ref{fig:detmodel} in Appendix~\ref{app:ratelossanalysis}), $F_0=(1-P) \Pi_0 + (1-P)(1-\eta)\Pi_1 + (1-P)(1-\eta)^2\Pi_2$ and $F_1=\mathbb{I} - F_0$, with $\Pi_i = |i\rangle\langle i|$, $i = 0, 1, 2$ being projectors corresponding to the vacuum, single photon and two photon outcomes of an ideal photon-number-resolving measurement, and, (iii) carrying out the mathematics of the linear-optic BSM operation on $\rho_i^{\otimes 2}$ while accounting for the appropriate post-selections as derived in Ref.~\cite{Lut99}.
\end{proof}

Once we have the state $\rho_i$, defined recursively in terms of $\rho_{i-1}$, we calculate the success probabilities, $P_s(i) = 4s_i$, where $s_i = s = a+b+2c$, $\forall i \ge 2$. The success probability of creating $\rho_1$, $P_s(1) = 1 - (1 - P_{s0})^M$, where $P_{s0} = 4s_1$, where $s_1 = a_e+b_e+2c_e$, is the probability of successful creation of an elementary link $\rho_1$ in one of the $M$ frequencies (see Appendix~\ref{app:success} for details).

We next prove that the sift-probability $P_1 = (q_1 + q_2 + q_3)^2, \forall i$, where $q_1 = (1-P_d)A_d$, $q_2 = (1-A_d)P_d$, and $q_3 = P_dA_d$, with $A_d = \eta_d + (1 - \eta_d)P_d$ (which are all functions of the loss and noise parameters of Alice's and Bob's detectors). An intuitive explanation is as follows: $q_2$ is the probability that the noisy detectors `flip' the outcome ($|10\rangle$ detected as (no-click, click), or $|01\rangle$ detected as (click, no-click)); $q_1$ is the probability that the detectors do not flip the outcome ($|01\rangle$ detected as (no-click, click), or $|10\rangle$ detected as (click, no-click)); and $q_3$ is the probability that the detectors generate the (click, click) outcome (regardless of whether $|10\rangle$ or $|01\rangle$ are detected. Since the flip, no-flip, and click-click probabilities are symmetric in the inputs $|01\rangle$ and $|10\rangle$, and each half of $\rho_i$ has exactly one photon (in two modes), regardless of the relative fractions of $|01\rangle$ and $|10\rangle$ in Alice's and Bob's share of the joint state, the probability of a successful sift is the probability they both get one of the above three events, hence $P_1 = (q_1+q_2+q_3)^2$. See Appendix~\ref{app:sift} for a more detailed argument.

The final step is to obtain the error probability
\begin{eqnarray}
Q_i &=& \frac{1}{P_1} \biggl( {\rm Tr} \bigl[ \rho_i(M_{0101} + M_{1010} + \nonumber\\
&& \frac{1}{2}\left\{M_{1101}+M_{1110}+M_{0111}+M_{1011}+M_{1111}\right\} ) \bigr] \biggr) \nonumber
\end{eqnarray}
where $P_1 = {\rm Tr}[\rho_i(M_{0101} + M_{0110} + M_{1001} + M_{1010} + M_{1101} + M_{1110} + M_{0111} + M_{1011} + M_{1111})]$, and $M_{ijkl} \equiv F_i \otimes F_j \otimes F_k \otimes F_l$. It is simple to argue that $Q_i$ is a function only of $2c_i/s_i$ (see Appendix~\ref{app:QBER} for detailed proof). The intuitive argument is that a bit error only arises from $2c_i/s_i$, the fractional probability of the classical correlation terms in $\rho_i$, whereas $(a_i + b_i)$ is the sum fractional probability of the two Bell states $|M^+\rangle$ ($a_i$) and $|M^-\rangle$ ($b_i$), with $s_i = (a_i+b_i) + 2c_i$. Even if the BSM results accidentally in a $|M^-\rangle$ to be formed, there would be no bit error. In order to calculate $c_i$, we calculate $c_i + d_i \equiv y_i$ and $c_i - d_i \equiv u_i$ by adding and subtracting Eqs.~\eqref{eq:cip1} and~\eqref{eq:dip1}, and writing recursions for $y_i$ and $u_i$. The solution to $y_i$ comes out as, $y_i = (s_i - z_i)/2$, with $z_i = (s^2/(a+b))((1+2w_1)(1+2w_r))^{-2^{i-1}}$, where $w_r = c/(a+b)$ and $w_1 = c_e/(a_e+b_e)$. The solution to $u_i$ requires us to solve the following variant of the chaotic {\em logistic map}: $w_{i+1} = w_r + 2(1-2w_r)w_i(1-w_i)$, where $w_i = u_i/z_i$. We derive the exact solution of this quadratic recursion (see Appendix~\ref{app:logisticmap} for proof), and are thus able to evaluate $Q_i = [1 - t_d(1 - 2c_i/s_i)]/2$, which simplifies to the form shown in Eq.~\eqref{eq:Q_exactform} of Theorem~\ref{thm:Q_main}.
\end{proof}

It is easy to account for a probabilistic entanglement source to account for a finite probability of vacuum in each time slot (the numerical calculations in Section~\ref{sec:numerics} further accounts for a non-zero two-pair generation probability). Such a probabilistic entanglement source can be modeled as generating $\rho = (1-p)|{\boldsymbol 0}\rangle\langle{\boldsymbol 0}| + p|M^{\pm}\rangle\langle M^{\pm}|$ in each frequency mode and in every $T_q$ second slot. Since $\rho$ can be regarded as the quantum state obtained by passing $|M^\pm\rangle$ through a beamsplitter of transmittance $p$, we can `push' $p$ through the BSM at the centers of elementary links, and apply our formulas after replacing $\lambda\eta_e$ by $\lambda\eta_e p$, and accordingly modifying the parameters: $a_e, b_e$, and $c_e$.

Finally, even though all the above analysis was done for $N=2^n$ elementary links (with $n$ an integer), we believe that the final formula for $Q$ and rate also hold for any integer $N$. In other words, with an end-to-end optical fiber channel with $N$ elementary links, $N \in {\mathbb Z}^+$,
\begin{equation}\label{eq:rate_general}
R_N(L) = {P_1P_{\rm succ}R_2(Q(N))}/{2T_q}\;{\text{key bits/s}},
\end{equation}
where,
$
Q(N) = \frac{1}{2}\left[1 - ({t_d}/{t_r})\left(t_rt_e\right)^{N}\right].
$
Since $R(Q) = 1-2h_2(Q)$, with $h_2(x) = -x\log_2(x) - (1-x)\log_2(1-x)$ the binary entropy function, the maximum range for which QKD is possible at a non-zero rate is determined by when $Q(N)$ exceeds $Q_{\rm th}$, where $h_2(Q_{\rm th}) = 1/2$ and $Q_{\rm th} \approx 0.1104$. One can invert $Q(N)$ to derive the maximum range as a function of number of elementary links $N$, and all the detector loss and noise parameters:
\begin{eqnarray}
&&L_{\rm max} = \left(\frac{20N}{\alpha}\right) \times \nonumber \\
&&\log_{10}\left[\frac{{\eta_e}\left(\sqrt{2(1-2P_e)H} - 2(1-2P_e)\right)}{4P_e}\right], \label{eq:Lmax}
\end{eqnarray}
where
$
H = 1+{t_r}/{\left[(1-2Q_{\rm th})\frac{t_r}{t_d}\right]^{\frac{1}{N}}}
$
and $\alpha$ is the fiber's loss coefficient, expressed in dB/km units.

\subsection{Rate-vs.-loss performance of the repeater chain}\label{sec:ratelossenvelope}

We defined $R_N(L)$ to be the secret key rate achievable with $N$ equal-length elementary links dividing up the total range $L$. Let us define $R_N^{(0)}(L)$ to be the secret key rate achieved with all the dark click probabilities set to zero, i.e., $P_e = P_r = P_d = 0$. It is reasonable to expect that non-zero dark click probabilities can only decrease the secret key rate (See Appexdix~\ref{app:ratelossenvelope1} for a more detailed discussion), and hence, $R_N(L) \le R_N^{(0)}(L)$. Assuming this to be true, the secret-key rate $R_N(L)$ can be upper bounded, to a very good approximation, by a three-segment rate plot (see Fig.~\ref{fig:rateloss_threepiece}): a constant-rate segment, a linear rate-vs.-transmittance segment, and a zero-rate segment. More specifically, we prove that:

\begin{figure}
\centering
\includegraphics[width=\columnwidth]{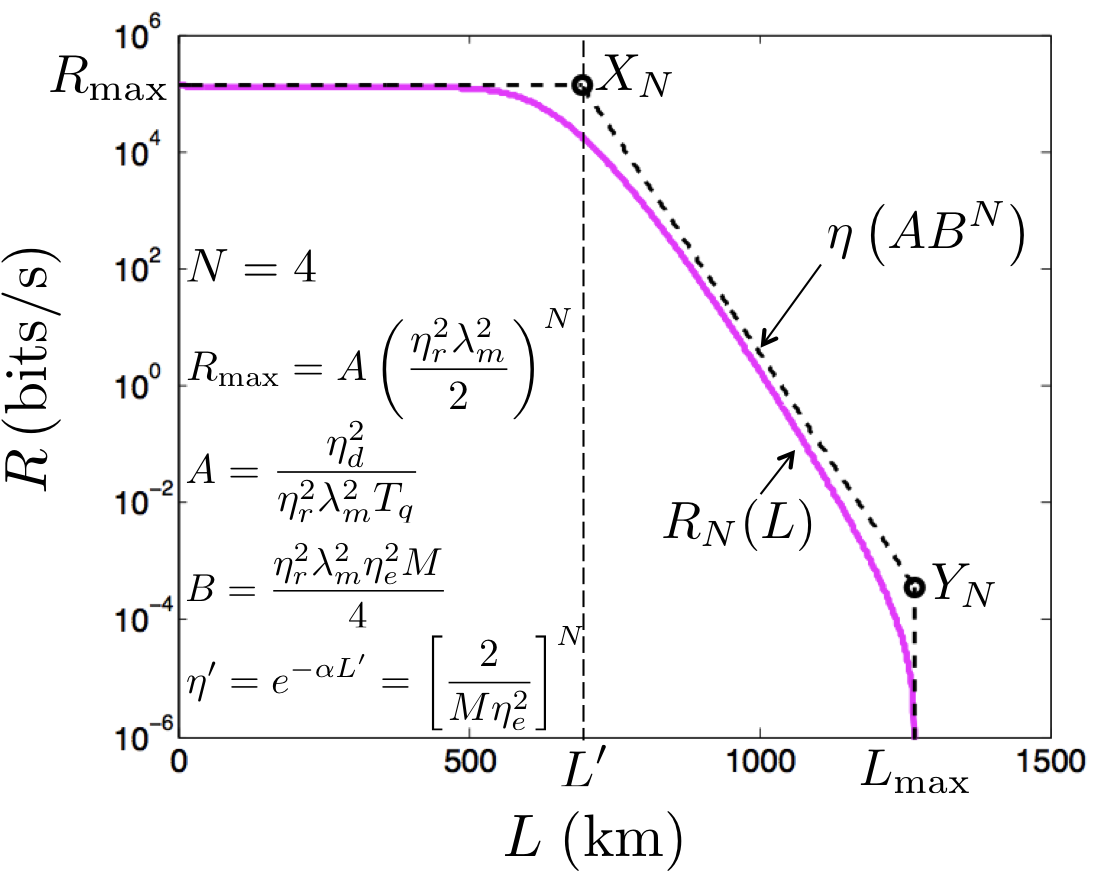}
\caption{(Color online) A `three-piece' upper bound to the rate-vs.-distance $R_N(L)$ achieved by a repeater chain consisting of $N$ elementary links over the range $L$ km. We assume the following parameters: $N=4$, $P_e=P_r=P_d = 3 \times 10^{-5}$, $\eta_e = \eta_r = \eta_d = 0.9$, $M = 1000$, $\lambda_m \equiv 1$ dB, $\alpha \equiv 0.15$ dB/km, $T_q = 50$ ns.}
\label{fig:rateloss_threepiece}
\end{figure}
\begin{theorem}\label{thm:ratelossUB}
The rate-vs.-distance function $R_N(L)$, achieved by a repeater chain comprising $N$ equal-length elementary links, can be upper bounded as:
\begin{equation}
R_N(L) \le R_N^{({\rm UB})}(L) = \left\{
\begin{array}{ll}
R_{\rm max}, & {\text{for}}\; 0 \le L \le L^\prime, \\
\eta \left(AB^N\right), & {\text{for}}\; L^\prime < L < L_{\rm max}, \\
0, & {\text{for}}\; L \ge L_{\rm max},
\end{array}
\right.\label{eq:ratelossUB}
\end{equation}
with $L^\prime = -\log(\eta^\prime)/\alpha$, $\eta^\prime = (2/M\eta_e^2)^N$, and $R_{\rm max} = A\,(\eta_r^2 \lambda_m^2/2)^N$, where the constants $A$ and $B$ are given by, $A = \eta_d^2/(\eta_r^2 \lambda_m^2 T_q)$ and $B = \eta_r^2 \lambda_m^2 \eta_e^2 M/4$, assuming non-zero detector dark-click probabilities cannot improve the key rate achievable by this repeater protocol, i.e., $R_N(L) \le R_N^{(0)}(L)$.
\end{theorem}
\begin{proof}
See Appendix~\ref{app:ratelossenvelope1}. The proof proceeds by upper bounding $R_N^{(0)}(L)$ individually by $R_{\rm max}$ and by $\eta \left(AB^N\right)$. The third segment is trivial since $R_N(L) = 0$ for $L \ge L_{\rm max}$, as we showed earlier.
\end{proof}

The third segment in Eq.~\eqref{eq:ratelossUB} disappears when $P_e = P_r = P_d = 0$, since $L_{\rm max} \to \infty$. It is straightforward to solve for the envelope of the points $\left\{X_N\right\}$, $N = 1, 2, \ldots$, where the first two segments of $R_N^{(\rm UB)}(L)$ intersect (see Fig.~\ref{fig:rateloss_threepiece}), and to prove that this envelope $R^{(\rm UB)}(L)$, is an upper bound to the actual rate-loss envelope $R(L)$:

\begin{theorem}\label{thm:ratelossUB_exp}
Assuming $R_N(L) \le R_N^{(0)}(L)$ holds for all $N \ge 1$, the rate-vs.-distance function $R(L)$ achieved by the repeater chain, once optimized over the choice of the number of elementary links $N$ as a function of the range $L$, can be upper bounded as:
\begin{equation}
R(L) \le R^{(\rm UB)}(L) = A\,\eta^t,
\end{equation}
where the power-law exponent $t$ is given by,
\begin{equation}
t = \frac{\log\left(\eta_r^2 \lambda_m^2/2\right)}{\log\left(2/M\eta_e^2\right)}.
\end{equation}
\end{theorem}

\begin{proof}
See Appendix~\ref{app:ratelossenvelope2} for the proof. We first show that $R_N(L) \le R_N^{(0)}(L), \forall N$ implies $R(L) \le R^{(0)}(L)$, where $R^{(0)}(L)$ is the overall rate-distance envelope, when $P_e = P_r = P_d = 0$. We then derive an upper bound to $R^{(0)}(L)$ by using the result in Theorem~\ref{thm:ratelossUB}.
\end{proof}

The above upper bound already suggests a power-law scaling of the true rate-loss envelope $R(L)$. It is actually possible to derive the zero-dark-click-probability rate-distance envelope $R^{(0)}(L)$ {\em exactly}, and as we show next, it is indeed given by a power law in the total Alice-to-Bob channel transmittance, $\eta$.

\begin{theorem}\label{thm:xiformula}
The rate-vs.-distance $R^{(0)}(L)$ achieved by a repeater chain when all detector dark-click probabilities are zero and an appropriate number of elementary links are used for a given range $L$, is exactly given by:
\begin{equation}
R^{(0)}(L) = A\,\eta^\xi,
\end{equation}
where $A = \eta_d^2/(\eta_r^2 \lambda_m^2 T_q)$, and the exponent $\xi$ is given by:
\begin{equation}
\xi = \frac{\log\left[\beta \left(1 - (1-\gamma z)^M\right)\right]}{\log z},
\label{eq:xi}
\end{equation}
where $z$ is the unique solution of the following transcendental equation in the interval $(0, 1)$:
\begin{eqnarray}
&&\left(1-(1-\gamma z)^M\right)\log \left[\beta(1 - (1-\gamma z)^M)\right] \nonumber\\
&&= \gamma Mz\log z \left(1-\gamma z\right)^{M-1},
\end{eqnarray}
with, $\beta = \eta_r^2\lambda_m^2/2$, and $\gamma = \eta_e^2/2$.
\end{theorem}

\begin{proof}
See Appendix~\ref{app:ratelossenvelope3}.
\end{proof}

\begin{figure}
\centering
\includegraphics[width=\columnwidth]{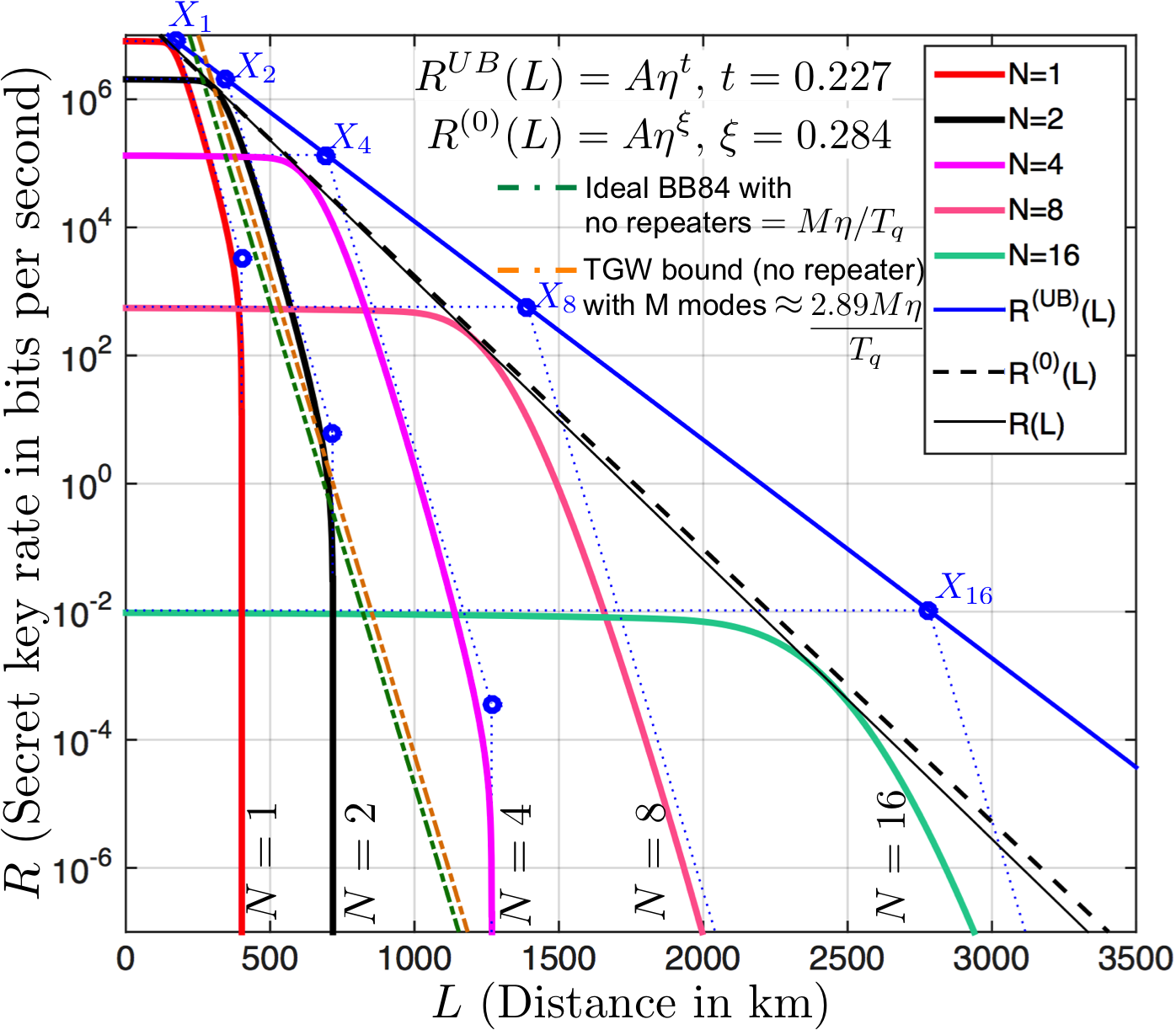}
\caption{(Color online) Secret key rates $R_N(L)$ as a function of range $L$ for $N = 1$, $2$, $4$, $8$ and $16$ elementary links. The rate-distance envelope is seen to outperform what is theoretically achievable by any repeater-less QKD protocol that uses the same time-slot length ($T_q$) and number of frequency channels ($M$), for $L \gtrsim 260$ km. The figure also shows the {\em exact} zero-dark-click-probability rate-distance envelope, $R^{(0)}(L) = A\eta^{\xi}$, where $\xi = 0.284$ (black dashed line). The envelope of the three-piece rate-distance upper bounds, $R^{\rm UB}(L) = A\eta^t$ is also shown (solid blue line), where $t = \log(\eta_r^2 \lambda_m^2/2)/\log(2/M\eta_e^2) = 0.227$. The parameters used are: $P_d = P_r = P_e = 3 \times 10^{-5}$, $\eta_d = \eta_r = \eta_e = 0.9$, $\lambda_m = 1$ dB (memory loss), $M=1000$ (frequency modes), $\alpha = 0.15$ dB/km (fiber loss), $T_q = 50$ ns.}
\label{fig:ratedistance_ideal}
\end{figure}
\begin{figure}
\centering
\includegraphics[width=\columnwidth]{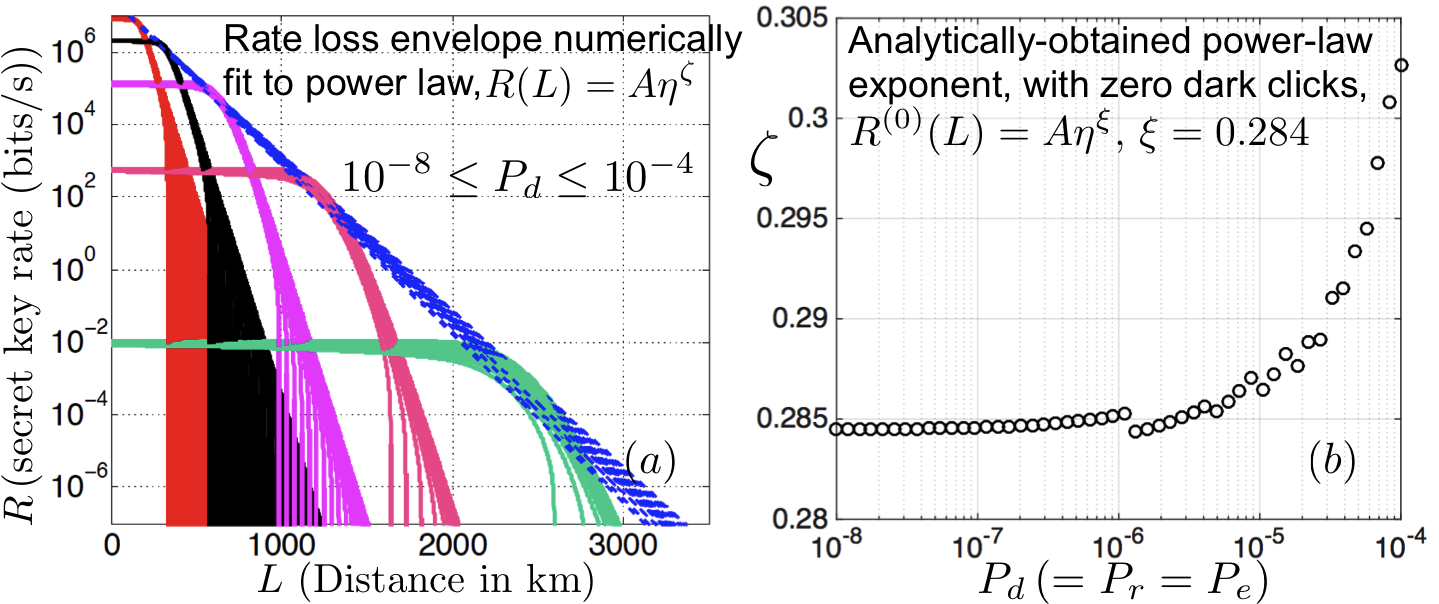}
\caption{(Color online) This figure captures the effect of detector dark click probability on the rate-loss scaling. It is seen that, for a given number of elementary links $N$, increasing the dark click probabilities drastically reduces the maximum range $L_{\rm max}$, however, the overall rate-distance envelope of the repeater chain remains largely unaffected over a significant, and practically feasible, range of detector dark click probability values (see the rate-loss-envelope traces as $P_d$ is varied from $10^{-8}$ to $10^{-4}$).}
\label{fig:exp_vs_Pd}
\end{figure}

In Fig.~\ref{fig:ratedistance_ideal}, we plot $R_N(L)$ as a function of $L$ for $N=2^n$ elementary links, with $n=0, 1, 2, 3, 4$. All the system parameters (listed in the figure caption) are kept the same for each plot. We also plot the three-piece upper bounds $R_N^{({\rm UB})}(L)$ (dotted blue lines), the envelope of those upper bounds $R^{({\rm UB})}(L) = A\eta^t$ (solid blue line), the rate-loss envelope $R^{(0)}(L) = A\eta^\xi$ with all detector dark-click probabilities set to zero (black dashed line), and the true (numerically-evaluated) rate-loss envelope $R(L)$ (black thin solid line). Fig.~\ref{fig:ratedistance_ideal} also shows the TGW bound corresponding to using all $M$ frequency modes (dash-dotted orange line) and the rate obtained by an ideal parallel BB84 implementation (perfect single-photon sources, and detectors) over all $M$ modes, $R = \eta M/T_q$ bits/s (dash-dotted green line). These two plots show that this repeater protocol's rate-loss performance fundamentally outperforms what is achievable without the assistance of quantum repeaters. Following are the main observations from Figs.~\ref{fig:ratedistance_ideal}, and~\ref{fig:exp_vs_Pd}:

{\em Effect of losses to the rate-loss envelope}---As noted in Theorem~\ref{thm:xiformula}, the exact power-law exponent $\xi$ of the true zero-dark-click-rate rate-loss envelope $R^{(0)}(L)$ has a complicated dependence on the system's loss parameters. On the other hand, the rate-loss envelope of the 3-piece upper bounds to $R_N(L)$ has a simple expression, $R^{(\rm UB)}(L) = A\eta^t$, with $A = \frac{\eta_d^2}{\eta_r^2 \lambda_m^2 T_q}$ and $t = \frac{\log\left(\eta_r^2 \lambda_m^2/2\right)}{\log\left(2/M\eta_e^2\right)}$, which makes its exponent $t$ useful to study the effects of various losses in the absence of dark clicks. Note that both the numerator and denominator in the expression for $t$ are negative for typical parameters. When the efficiency of the repeater node $\eta_r \lambda_m$ decreases, $t$ increases (thus making the rate-loss scaling worse; $t=1$ being the TGW limit, performance attainable without repeaters). Note that $(\eta_r \lambda_m)^2$ can be roughly interpreted as the probability of success (for the two memories and two detectors) at a repeater node. On the other hand, $M \eta_e^2$ can be roughly interpreted as the probability of success (for at least one of $M$ spectral modes and the two detectors) at the center of an elementary link. When $M \eta_e^2$ increases, $t$ decreases (thus making the rate-loss scaling better). Finally, note that the efficiency of Alice's and Bob's detectors $\eta_d$ does not affect the rate-loss scaling, but $\eta_d^2$ is an overall multiplier to the rate via the pre-factor $A$ (as expected, due to a $\eta_d^2$ multiplicative reduction in the number of usable time slots for key generation).

{\em Effect of dark click probability}---To examine the effect of detector dark-click probabilities to the secret key rate, we set $P_d = P_r = P_e$. The effect of $P_d$ to $R_N(L)$ is captured primarily by the maximum range $L_{\rm max}$, i.e., the third segment of $R_N^{(\rm UB)}(L)$ in Eq.~\eqref{eq:ratelossUB}. The envelope of the three-piece upper bounds, $R^{(\rm UB)}(L)$, is however completely unaffected by $P_d$, since the envelope is the locus of the corner-points $\left\{X_N\right\}$, $N = 1, 2, \ldots$, while being unaffected by the corner-points $\left\{Y_N\right\}$. We numerically fit the exact rate-distance envelope to the power law $R(L) = A\eta^\zeta$, and show that the exponent $\zeta$ remains largely unaffected over a significant (and practically feasible) range of $P_d$ (see Fig.~\ref{fig:exp_vs_Pd}(a)). In other words, $\zeta(P_d) \approx \xi$, the exact power-law exponent when $P_d = 0$, given in Eq.~\eqref{eq:xi}, over a significant range of $P_d$ (see Fig.~\ref{fig:exp_vs_Pd}(b)). The maximum range $L_{\rm max}$ achieved by a given number of elementary links $N$, however, drastically decreases with increasing $P_d$ (see Fig.~\ref{fig:exp_vs_Pd}(a)). In the regime that $P_d \ll 1$ and the deviations from ideal detection efficiency ($\epsilon_d \equiv 1-\eta_d$) and memory efficiency ($\epsilon_r \equiv 1-\eta_r \lambda_m$) are small, one can show that, to first order in $P_d, \epsilon_d, \epsilon_r$, we have $t_r \approx t_r/t_d \approx 1-4P_d$. This yields a simpler expression for the maximum range, $L_\mathrm{max} \approx  (20 N/ \alpha) \log_{10} \left[\left(\sqrt{2 (1+(1-2 Q_{\rm th})^{-1/N})}-2\right)/4 P_e \right]$, which shows that the first-order dependence of $L_{\rm max}$ to detector dark clicks is via a subtractive term, $-(20N/\alpha)\log_{10}(4P_e)$, which makes $L_{\rm max}$ to go to infinity as $P_e \to 0$, as expected.

{\em Optimal choice of the number of repeaters}---For a given Alice-to-Bob range $L$, it should be divided up into an optimum number of equal-length elementary links, in order to maximize the key rate. At a short range, using too many repeaters diminishes the end-to-end key rate, due to the $50\%$ heralding efficiencies of the linear-optic BSMs at the repeater nodes. Employing higher-efficiency BSMs (by injecting ancilla single photons for instance~\cite{Ewe14}) will increase $R_{\rm max}$ in Fig.~\ref{fig:rateloss_threepiece}, and will hence increase $N^*(L)$ at any given range $L$.

{\em Beating the TGW bound}---The secret key rate of any QKD protocol that does not use quantum repeaters is upper bounded by the TGW bound, $R^{(\rm UB)}_{\rm TGW}(\eta) = \log\left((1+\eta)/(1-\eta)\right)$ bits per mode~\cite{Tak13}, $\eta$ being the total channel transmittance. $R^{(\rm UB)}_{\rm TGW}(\eta) \approx 2.88\eta$, when $\eta \ll 1$ (high loss). The BB84 protocol---both the single-photon based and the weak coherent state implementation employing decoy states---as well as continuous-variable (CV) QKD with a Gaussian input modulation, attain key rates, $R \approx \eta$ bits/mode~\cite{Sca09}, thereby leaving little room for improvement by any other protocol. With $M$ orthogonal frequency channels available, and a qubit duration of $T_q$ seconds, a parallel implementation of an ideal QKD protocol on each of those frequency channels cannot exceed a key rate of $MR^{(\rm UB)}_{\rm TGW}(\eta)/T_q$ bits/s, a plot shown in Fig.~\ref{fig:ratedistance_ideal} (see dash-dotted orange line). The rate-loss function $R(L)$ attained by our repeater architecture distinctly outperforms this fundamental repeater-less rate-loss limit, as is also clear from the power law dependence $R(L) = A\eta^\zeta$ with $\zeta < 1$, whereas the TGW limit corresponds to $\zeta = 1$.

{\em Choice of the number of frequency modes}---An important part of the design of the repeater architecture is choosing $M$, the number of frequency modes that the elementary links use for multiplexing. In Fig.~\ref{fig:xi_vs_M}, we plot the power law exponent $\xi$ of the zero-dark-click rate-loss envelope $R^{(0)}(L) = A\eta^\xi$, as a function of $M$. In order to obtain a desired performance improvement over the TGW bound's scaling limit (i.e., $\xi = 1$), the lower the detector efficiencies $\eta_e$ and $\eta_r$, the higher is the level of frequency multiplexing needed. Note that $\xi$ does not depend upon the efficiency $\eta_d$ of Alice's and Bob's detectors (see Theorem~\ref{thm:xiformula}). Furthermore, as is intuitively clear, and apparent from comparing the $\xi(M)$ plots for $\eta_e =0.5, \eta_r = 0.9$ and $\eta_e =0.9, \eta_r = 0.5$, that it is more important for the repeater-node detectors to have high efficiency as compared to the detectors at the middle of the elementary links, since frequency multiplexing ``helps" the latter detectors. Next, we note that there is a minimum number of frequency modes $M_{\rm min}$ needed for this repeater protocol to be useful (i.e., barely beat the TGW bound's scaling limit), which increases as $\eta_e$ and $\eta_r$ decrease. An interesting, yet intuitive thing to note, is that the blue solid and the black dashed (as well as the red diamonds and the magenta dash-dotted) curves pairwise come close to one another as $M$ increases. This happens because when $M$ becomes sufficiently large, the probability of successful creation of an elementary link $P_s(1) = 1 - (1 - P_{s_0})^M\approx 1$, which has a weak dependence on $\eta_e$, and hence $\xi$ depends more strongly on the losses at the repeater nodes, i.e., $\eta_r$. The exact expression for the power-law exponent of $R^{(\rm UB)}(L)$ --- which is a lower bound to the true exponent $\xi$, i.e., $t = \frac{\log\left(\eta_r^2 \lambda_m^2/2\right)}{\log\left(2/M\eta_e^2\right)} = \frac{1 + 2\log_2(1/\eta_r \lambda_m)}{\log_2 M - \left[1 + 2\log_2 (1/\eta_e)\right]} < \xi$ --- provides a useful guideline for the choice of $M$, as well as illustrates the aforesaid effect (of the dependence of the power-law exponent being primarily on $\eta_r$ when $M$ is high enough).
\begin{figure}
\centering
\includegraphics[width=0.7\columnwidth]{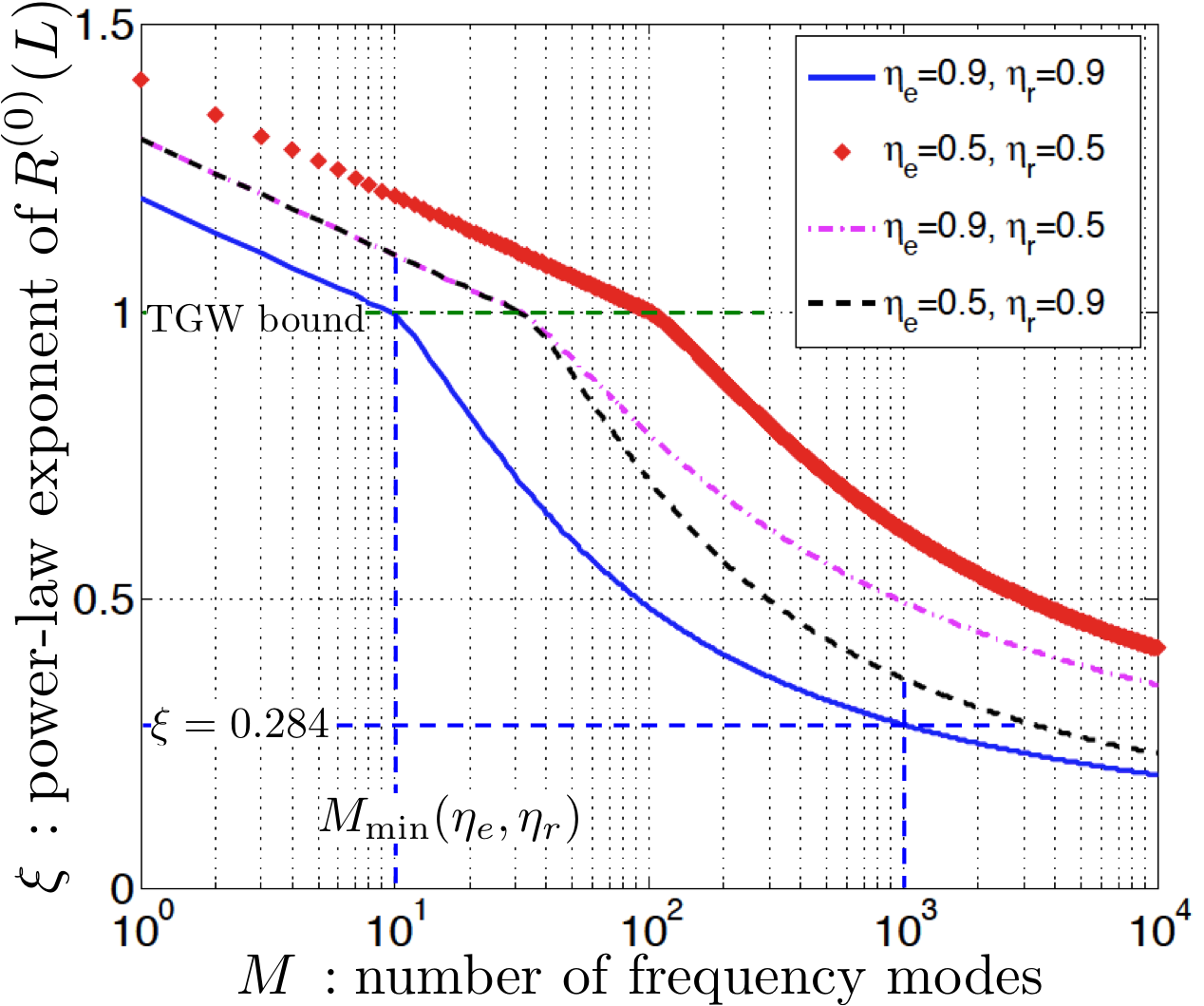}
\caption{(Color online) Here we plot the power law exponent $\xi$ of the zero-dark-click rate-loss envelope $R^{(0)}(L) = A\eta^\xi$, as a function of the number of frequency modes $M$. In order to obtain a desired performance improvement over the TGW rate-loss scaling ($\xi = 1$), lower are the detector efficiencies $\eta_e$ and $\eta_r$, higher is the level of frequency multiplexing needed.}
\label{fig:xi_vs_M}
\end{figure}

\subsection{Entanglement distillation rates}\label{sec:other rates}

The actual end to end shared quantum state after successfully connecting $2^{i-1}$ elementary links is given by (see Appendix~\ref{app:QBER_state} for proof):
\begin{eqnarray}
\rho_i &=& \frac{1}{s_i}\left[a_i|M^+\rangle\langle M^+| + b_i|M^-\rangle\langle M^-| + c_i|\psi_0\rangle\langle\psi_0| \right.\nonumber\\
&&\left. +\; d_i|\psi_1\rangle\langle\psi_1| + d_i|\psi_2\rangle\langle\psi_2| + c_i|\psi_3\rangle\langle\psi_3|\right],\label{eq:stateexpression_maintext}
\end{eqnarray}
where $|\psi_0\rangle = |01,01\rangle$, $|\psi_1\rangle = |01,10\rangle$, $|\psi_2\rangle = |10,01\rangle$, $|\psi_3\rangle = |10,10\rangle$, $|M^{\pm}\rangle = \left[|\psi_2\rangle \pm |\psi_1\rangle\right]/\sqrt{2}$, $s_i = a_i + b_i + 2(c_i + d_i)$, and the coefficients given as:
\begin{eqnarray}
a_i &=& \frac{1}{2}\left[1 + \left(\frac{a-b}{a+b}\right)^{i-1}\left(\frac{a_e-b_e}{a_e+b_e}\right)\right]z_i, \nonumber\\
b_i &=& \frac{1}{2}\left[1 - \left(\frac{a-b}{a+b}\right)^{i-1}\left(\frac{a_e-b_e}{a_e+b_e}\right)\right]z_i, \nonumber\\
c_i &=& \frac{s_i}{4}\left[1 - \frac{z_i}{s_i(1-2w_r)} \left[(1-2w_r)(1-2w_1)\right]^{2^{i-1}}\right], \nonumber\\
d_i &=& \frac{s_i}{4} - \frac{z_i}{2}\left[1 - \frac{1}{2(1-2w_r)}\left[(1-2w_r)(1-2w_1)\right]^{2^{i-1}}\right],\nonumber
\end{eqnarray}
with $w_1 = c_e/(a_e+b_e)$, $w_r = c/(a+b)$, $s_1 = a_e+b_e+2c_e$, $s_i = s = a+b+2c$, $2 \le i \le n+1$, and $z_i$ given by,
\begin{equation}
z_i = \left(\frac{s^2}{a+b}\right)\left(\frac{1}{(1+2w_1)(1+2w_r)}\right)^{2^{i-1}}, \, i \ge 2,
\end{equation}
with $z_1 = a_e + b_e$. The expressions for $a_i$, $b_i$, $c_i$, and $d_i$ reduce to $a_e$, $b_e$, $c_e$, and $0$, respectively, for $i=1$.

\begin{figure}
\centering
\includegraphics[width=0.9\columnwidth]{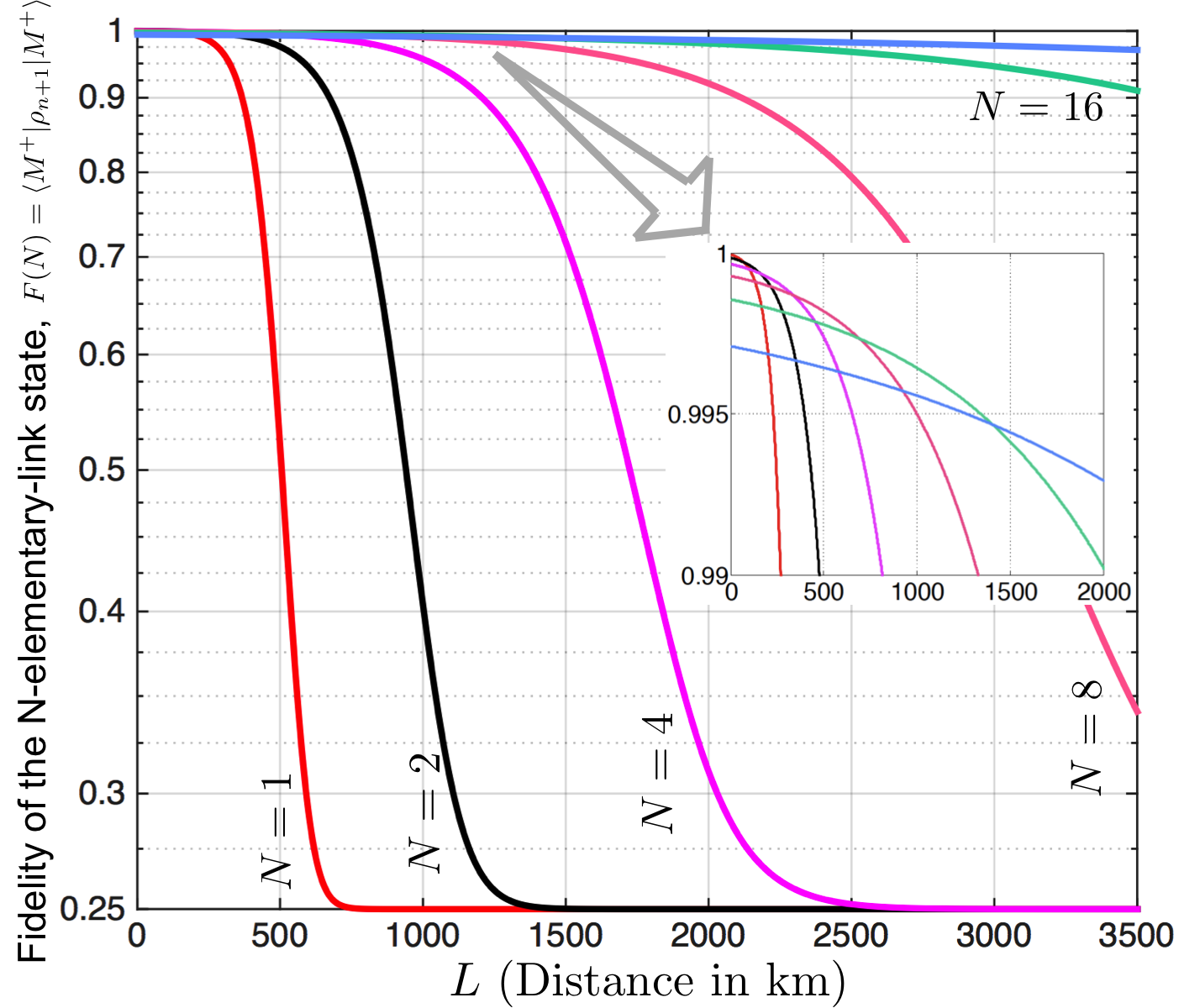}
\caption{(Color online) Fidelity of the $2^n$-link state, $\rho_{n+1}$ for $2^n = N = 1, 2, 4, 8, 16, 32$, with respect to the pure EPR state $|M^+\rangle$. We used $P_e=P_r=P_d = 3 \times 10^{-5}$, $\eta_e = \eta_r = \eta_d = 0.9$, $M = 1000$, $\lambda_m \equiv 1$ dB, and $\alpha \equiv 0.15$ dB/km.}
\label{fig:fidelity}
\end{figure}
The fidelity of the $N=2^n$ elementary link state $\rho_{n+1} \equiv \rho^{AB}(N)$ with respect to $|M^+\rangle$, ${\langle M^+ |\rho_i |M^+\rangle}$,
\begin{equation}
F_N(L) = {(a_{n+1}+d_{n+1})/s}.
\end{equation}
In Fig.~\ref{fig:fidelity}, we plot $F_N(L)$ as a function of the range $L$ for $N=1, 2, 4, \ldots, 32$ elementary link concatenations. Note that the plots show the fidelity of the actual heralded state (the probability $P_{\rm succ}$ of generating $\rho^{AB}_{n+1}$ successfully is not being accounted for). It is seen that the maximum range $L_{\rm max}$ for the secret-key generation rate $R_N(L)$ roughly corresponds to a state fidelity of $F_N(L) \approx 0.86$ for all $N$.

If Alice and Bob have many copies of the state $\rho^{AB}$, with no restriction on their actual quantum measurements and post-processing, and only using one-way classical communication over the public channel, the rate at which they can generate shared entanglement $E_D(\rho^{AB})$---measured in ebits (clean EPR pairs) per copy of $\rho^{AB}$ initially shared---is lower bounded by the coherent information $I(A \rangle B) = H(B) - H(AB)$, also known as the {\em hashing bound}~\cite{Dev04}. The hashing bound for the $N$-link shared rate $\rho^{AB}(N)$ can be evaluated to yield:
\begin{equation}
I_N(A \rangle B) = 1-H\left(\frac{c_{n+1}}{s},\frac{c_{n+1}}{s},\frac{a_{n+1}+d_{n+1}}{s},\frac{b_{n+1}+d_{n+1}}{s}\right) \nonumber,
\end{equation}
where $H(\cdot)$ is the Shannon entropy function. Since $\rho^{AB}(N)$ is heralded with probability $P_{\rm succ}$, and since each qubit occupies $T_q$ seconds, the achievable entanglement-distillation rate is given by:
\begin{equation}
E_N(L) = P_{\rm succ} I_N(A \rangle B)/T_q,
\end{equation}
which is plotted in Fig.~\ref{fig:entanglement} for $N=1,2,\ldots, 16$. It is instructive to compare this with the expression for the secret-key-generation rate:
\begin{equation}
R_N(L) = P_1 P_{\rm succ} R_2(Q_{n+1})/2T_q,
\end{equation}
where $R_2(Q_{n+1}) = 1 - 2H(Q_{n+1}, 1-Q_{n+1})$. When $P_d=P_r=P_e=0$ (all detector dark click rates are zero), $a_i = a$, and $b_i = c_i = d_i = 0$, and therefore $\rho_i = |M^+\rangle \langle M^+ |$ for all $1 \le i \le n+1$. Thus the QBERs, $Q_i = 0$, resulting in $R_2(Q_{n+1})=1$, and $I_N(A \rangle B) = 1$. Therefore, $E_N(L)$ and $R_N(L)$ differ only by a factor of $P_1/2 = \eta_d^2/2$, as intuitively expected. Clearly, the same is true for the zero-dark-click rate-distance envelopes, $E^{(0)}(L)$ and $R^{(0)}(L)$, i.e., $E^{(0)}(L) = (2/\eta_r^2\lambda_m^2T_q)\,\eta^\xi$, where $\xi$ is given by Eq.~\eqref{eq:xi}. Similar to the secret-key-generation rates, when the dark click probabilities are non-zero (however small), there is a finite maximum range for entanglement distillation with $N$ links, but the rate-loss envelope $E(L)$ is only slightly affected. In Fig.~\ref{fig:entanglement}, we plot $E_N(L)$ for $N = 1, 2, \ldots, 16$ for $P_d=P_e=P_r = 3 \times 10^{-5}$, along with the zero-dark-click envelope $E^{(0)}(L)$, showing that the rate-distance envelope is practically the same for this dark click level.

The maximum range for secret-key generation results from the condition $R_2(Q_{n+1}) = 0$, which gives the expression for $L_{\rm max}^{\rm QKD}$ given in Eq.~\eqref{eq:Lmax}. The maximum range for entanglement distillation derives from the condition $I_N(A \rangle B) = 0$, i.e., $H\left(\frac{c_{n+1}}{s},\frac{c_{n+1}}{s},\frac{a_{n+1}+d_{n+1}}{s},\frac{b_{n+1}+d_{n+1}}{s}\right) = 1$. Unlike the key-generation rate, which depends cleanly on one parameter: the QBER, the entanglement distillation rate depends in a more complicated fashion on the shared state $\rho^{AB}_{n+1}$, through the parameters $a_{n+1}, b_{n+1}, c_{n+1}, d_{n+1}$, and hence an analytic formula for the maximum range $L_{\rm max}^{\rm ent-dist}$ is not possible to obtain. The maximum ranges for entanglement distillation, evaluated numerically, work out to be somewhat higher compared with the those for secret-key generation, for identical system parameters. For the parameters considered in Figs.~\ref{fig:ratedistance_ideal} and~\ref{fig:entanglement}, for $N = 1, 2, 4, 8, 16$, we get (rounded to a km):
\begin{eqnarray}
L_{\rm max}^{\rm QKD} &=& [401, \quad 716, \quad 1267, \quad 2208, \quad 3772], \\
L_{\rm max}^{\rm ent-dist} &=& [411, \quad 761, \quad 1389, \quad 2488, \quad 4367].
\end{eqnarray}

In evaluating the above range numbers for the QKD case, we assumed zero dark click rates for the Alice-Bob detectors (i.e., $P_d = 0$, $P_e = P_r = 3 \times 10^{-5}$), in order for an unbiased comparison, i.e., for both cases above, Alice and Bob start with many copies of the noisy EPR state $\rho_{n+1}^{AB}$. It is instructive to note that an achievable shared entanglement generation rate is automatically an achievable secret-key generation rate. Therefore, our results show that the QKD protocol we analyzed is (ever so slightly) suboptimal, in the sense that if Alice and Bob held many copies of the noisy EPR pairs $\rho_{n+1}^{AB}$ in perfect quantum memories, and applied an ideal entanglement distillation protocol~\cite{Dev04}, and then converted those EPR pairs to shared secret key bits, the resulting secret-key rates, and the maximum ranges would be slightly higher compared to what we got. It is remarkable however how close to that ultimate limit a QKD protocol even with a simple measurement and post-processing can get.

\begin{figure}[ht]
\centering
\includegraphics[width=0.9\columnwidth]{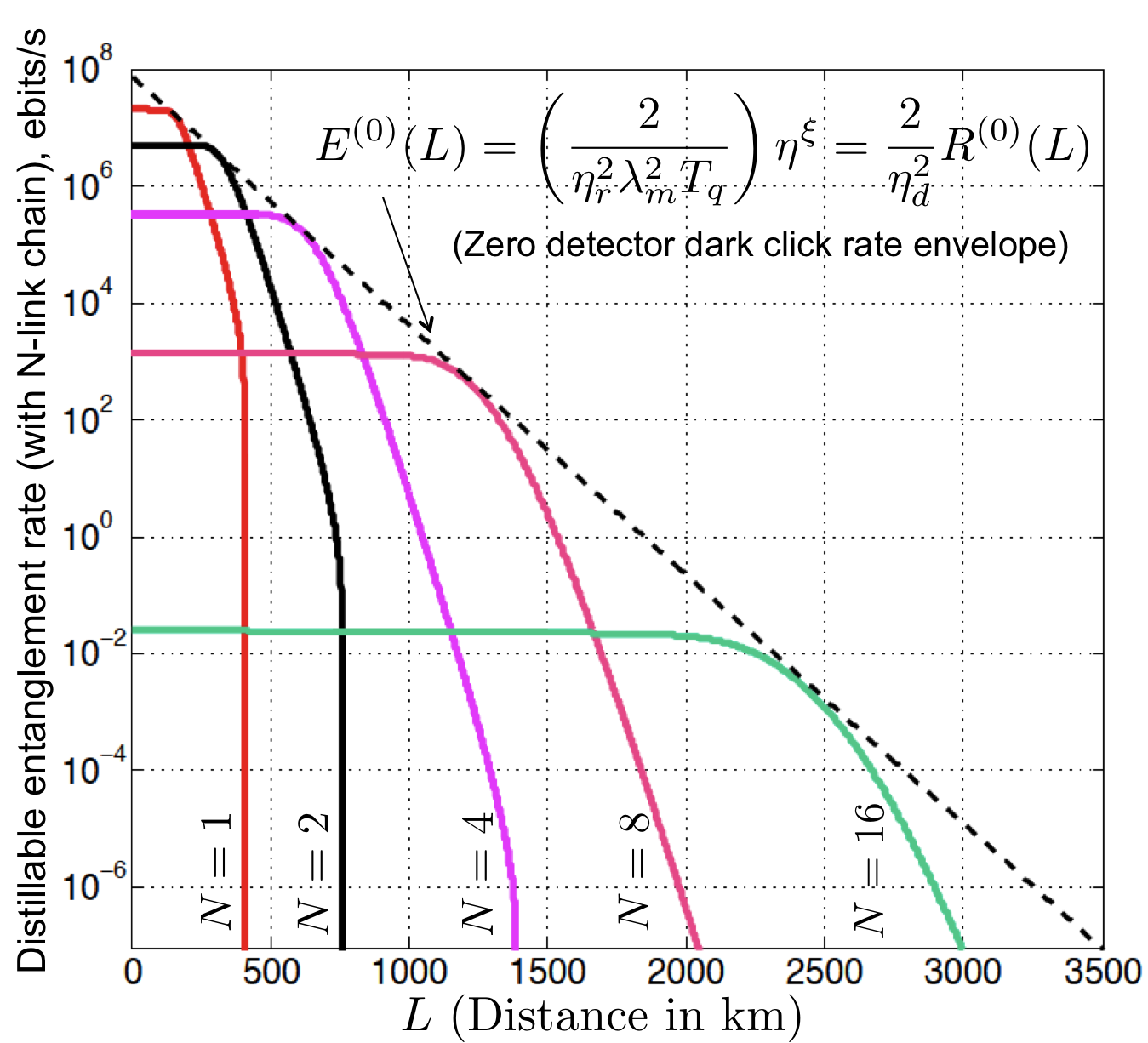}
\caption{(Color online) Achievable entanglement distillation rate (measured in pure EPR pairs per second) using an $N$-link repeater chain, for $N = 1, 2, 4, 8, 16, 32$. We used $P_e=P_r=P_d = 3 \times 10^{-5}$, $\eta_e = \eta_r = \eta_d = 0.9$, $M = 1000$, $\lambda_m \equiv 1$ dB, $\alpha \equiv 0.15$ dB/km, and $T_q = 50$ ns long qubits.}
\label{fig:entanglement}
\end{figure}

\section{The effect of two-pair emissions}\label{sec:numerics}

The entire theoretical analysis in Section~\ref{sec:analysis}, as well as all the calculations in the Appendices, assume that the entangled photon pair sources have a zero probability of multi-pair emission, which is usually not the case in practice, particularly when one employs spontaneous parametric downconversion (SPDC) to generate entangled pairs. The purpose of this section is to extend our analysis to sources whose two-pair probability, $p(2) > 0$. Even though one could in principle attempt a fully analytical calculation of the entangled state propagation through the repeater chain (along the lines of our derivations in Appendix~\ref{app:elem}), such a calculation would be extremely tedious. We instead set up an {\em exact} numerical calculation of the quantum states of the elementary link and the states resulting from successful BSM connections, where we evolve the quantum states in the Fock basis, and use the sparse matrix toolbox of MATLAB to create time-efficient subroutines for beamsplitters, partial trace operations, and photon-number-resolving detectors. We continue to assume however that all detectors in the system have single-photon resolution.

We use this numerical code to evaluate $R_N(L)$ for a particular form of source with $p(2) > 0$ (see Eq.~\eqref{eq:stateform}). We find that for a given $p(2)$, up to a certain maximum number of elementary links, the rate-distance performance remains almost identical to what is attained by an ideal ($p(2)=0$) source (i.e., that evaluated in Section~\ref{sec:analysis}). However, the rate becomes close-to-zero at any range, when $N \ge N_{\rm max}(p(2))$ (see Fig.~\ref{fig:p_2_non_zero}). Our numerical calculations also show that the scaling law in Eq.~\eqref{eq:errorpropagation} for error-propagation through the repeater chain continues to hold---with an appropriate $p(2)-dependent$ modification to the pre-factor $(t_r/t_d)$---even for non-ideal sources (see Fig.~\ref{fig:Q_vs_p2}).

This Section is organized as follows. In subsection~\ref{sec:ratelossp2}, we will show the empirical effect of $p(2)$ on the rate-loss behavior of the repeater architecture. In subsection~\ref{sec:numericalmodel}, we will develop a phenomenological model for QBER scaling  (an extension of Eq.~\eqref{eq:errorpropagation} when $p(2)>0$), which we will use in turn to develop an approximate model to understand the functional form of $N_{\rm max}(p(2))$.

\subsection{Rate-loss behavior with non-ideal sources}\label{sec:ratelossp2}

\begin{figure*}
\centering
\includegraphics[width=\textwidth]{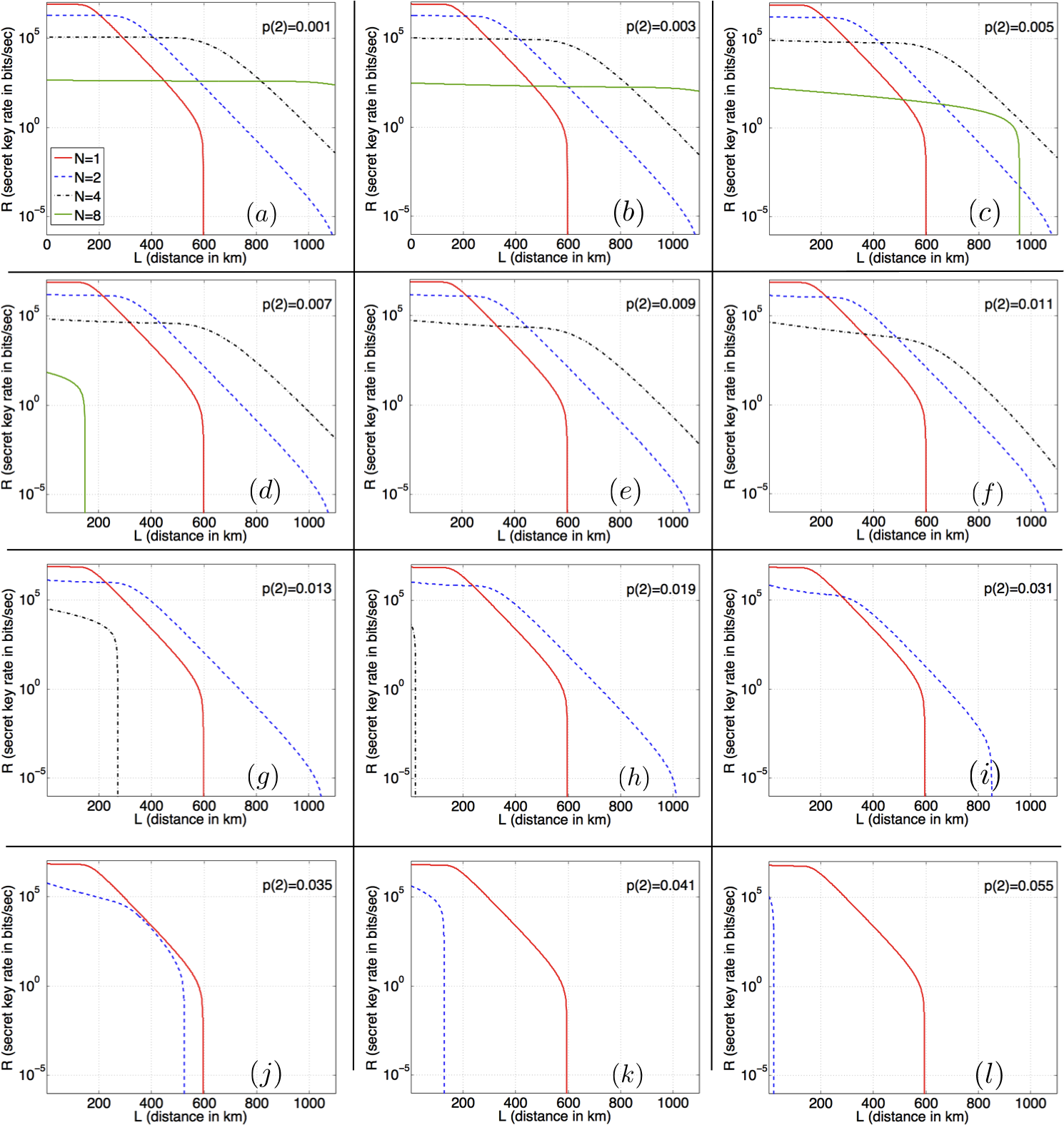}
\caption{(Color online) Secret key rate $R$ (bits/s) vs.~distance $L$ (km), evaluated for $n=0,1,2,3$ ($N=1,2,4,8$ elementary links), for sources with two-pair emission probability $p(2)$ ranging from $0.001$ to $0.055$. At any given value of $p(2)$, there is a certain number of elementary links up until which the rate-loss envelope achieved by the repeater architecture remains almost identical to what is achieved by a $p(2)=0$ entanglement source. However, as soon as $N \ge N_{\rm max}(p(2))$, the rate goes to zero an any range. The parameter values used are: $P_d = P_r = P_e = 10^{-6}$, $\eta_d = \eta_r = \eta_e = 0.9$, $\lambda_m = 1$ dB (memory loss), $M=1000$ (frequency modes), $\alpha = 0.15$ dB/km (fiber loss), and $T_q = 50$ ns. The plots show that, for these parameters, for $p(2)=0.035$, it is best to have a single elementary link between Alice and Bob over the entire range. The rate-loss tradeoff for 2 elementary links is worse at all range values. Similarly, at $p(2)=0.013$, using $4$ elementary links does not yield a better rate compared to what is attained with $2$ elementary links, at all range values. Interestingly however, the rate-distance plots come crashing down from higher $N$ values to lower (number of elementary links) {\em one at a time} as $p(2)$ is increased, while the rate-distance tradeoffs for the lower $N$ values stay almost at their $p(2)=0$ levels. Note that the $N=1$ plot has no perceivable change from $p(2)=0.001$ to $p(2)=0.055$. Similarly, the $N=2$ plot has no perceivable change from $p(2)=0.001$ to $p(2)=0.019$.}
\label{fig:p_2_non_zero}
\end{figure*}
In Fig.~\ref{fig:p_2_non_zero}, we plot the secret key rates $R_N(L)$ for $N=1, 2, 4, 8$ elementary links ($n=0, 1, 2, 3$) with all parameters held constant, $p(1) = 0.9$ and several choices of $p(2)$ ranging from $0.001$ to $0.015$. We model the non-ideal entanglement source as generating the state~\cite{Sli03},
\begin{eqnarray}\label{eq:stateform}
|\psi\rangle &=& \sqrt{1-p(1)-p(2)}\,|00,00\rangle + \sqrt{p(1)}\,|M^+\rangle \nonumber \\
&&+ \sqrt{p(2)/3}\,\left(|20,02\rangle - |11,11\rangle + |02,20\rangle\right),
\end{eqnarray}
where $|M^+\rangle = [|10,01\rangle + |01,10\rangle]/\sqrt{2}$. This particular form of the entangled photon-pair state, and in particular the form of the $4$-photon term, is motivated by parametric down-conversion sources~\cite{PRL187902}. If $p(2)$ is small, the exact form of the two-pair term does not seem to affect the results, notwithstanding that our simulation is easily able to take into account any particular form of the two-pair term, depending upon the physical model of the actual source of entanglement. Finally, we assume that the higher-order multi-pair emission terms ($3$-pair or higher) have significantly lower probabilities compared to the two-pair term, and that $p(2)$ effectively captures the effect of multi-pair emissions to the secret-key rates. One other difference in the rate-loss behavior compared with the $p(2)=0$ theoretical analysis in Section~\ref{sec:analysis} is that the QBER can be now non-zero even when the detector dark click rates are zero. This is because errors in the sifted bit may now be caused by the multi-pair events generated by the entanglement sources.

At a given $p(2)$, there is a maximum number of elementary links up until which the rate-loss envelope achieved by the repeater architecture remains almost identical to what is achieved by a $p(2)=0$ entanglement source. When $N \ge N_{\rm max}(p(2))$, the rate $R(L) = 0$, $\forall L \ge 0$. Seen differently, the rate-distance plots in Fig.~\ref{fig:p_2_non_zero} come crashing down from higher to lower values of $N$ values (number of elementary links) {\em one at a time} as $p(2)$ is increased from $0$ (with $p(1)=0.9$ held constant), while the rate-distance plots for the lower $N$ values stay unaffected, i.e., almost at its $p(2)=0$ level, until $p(2)$ becomes high enough to make the next lower value of $N$ unsustainable. As an example, the $N=1$ plot has no perceivable change from $p(2)=0.001$ to $p(2)=0.055$. Similarly, the $N=2$ plot has no perceivable change from $p(2)=0.001$ to $p(2)=0.019$.

\subsection{Phenomenological model for QBER scaling and maximum usable number of elementary links}\label{sec:numericalmodel}
\begin{figure}
\centering
\includegraphics[width=\columnwidth]{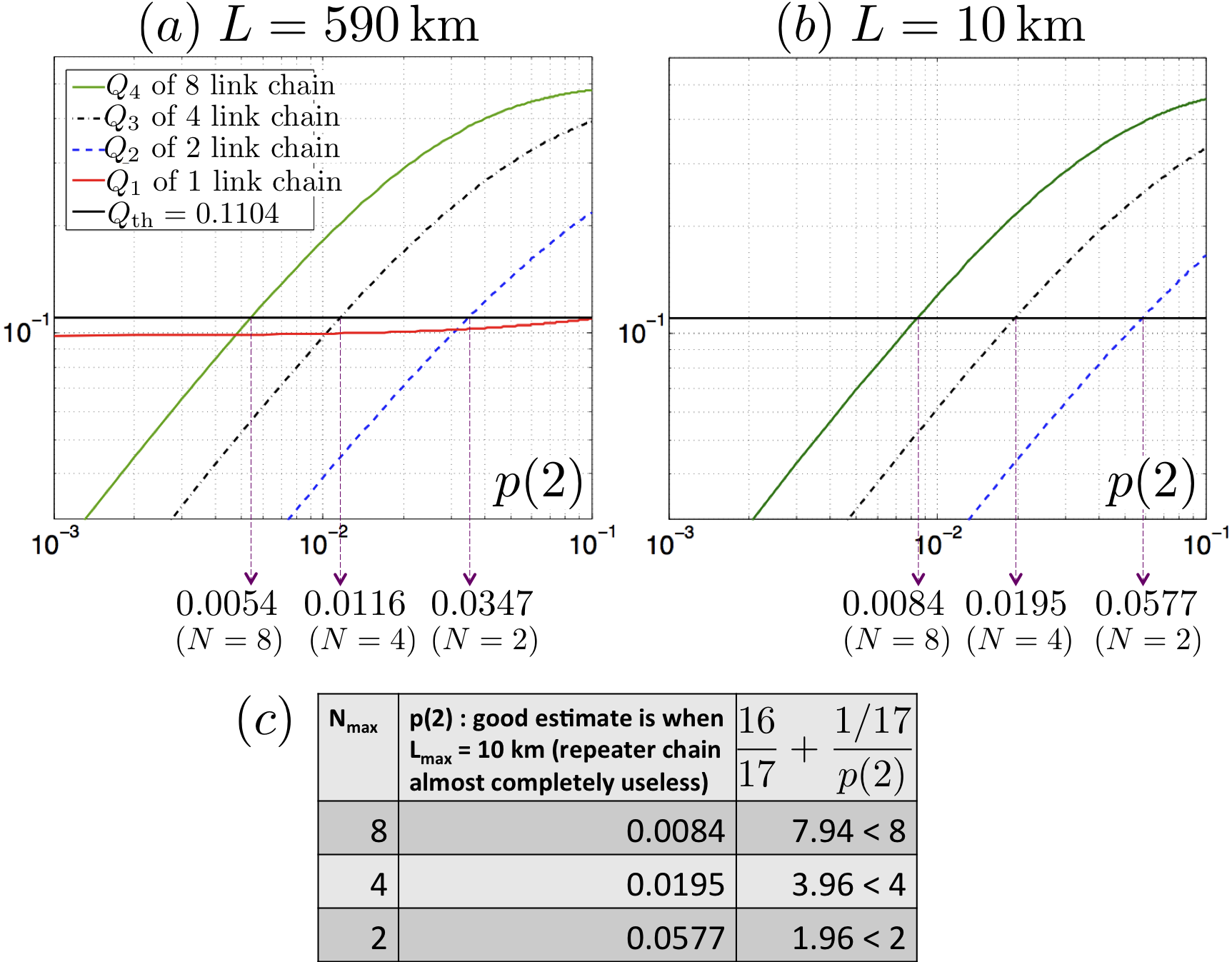}
\caption{(Color online) The purpose of this figure is to gauge the $p(2)$ values where a certain length $N$ (and higher) of the repeater chain becomes ineffective, as depicted in Fig.~\ref{fig:p_2_non_zero}. We choose two fixed {\em maximum} range values: one a number close to zero ($10$ km) to assess $N_{\rm max}(p(2))$, and the other a little below the range of a single elementary link ($590$ km). We divide an overall range $L$ of (a) $10$ km and (b) $590$ km, into $N = 1, 2, 4$ and $8$ elementary links, and plot the end-to-end QBER for each case, as a function of the two-pair-emission probability $p(2)$. The black horizontal line corresponds to $Q_{\rm th} = 0.1104$. The secret key rate goes to zero, when the end to end QBER exceeds $Q_{\rm th}$. It is instructive to tally the $p(2)$ values where $Q_{n+1}$ crosses the $Q_{\rm th}$ line for $n=0, 1, 2, 3$, with the plots in Fig.~\ref{fig:p_2_non_zero}. We assume, $\eta_e=\eta_r=\eta_d = 0.9$, $P_e = P_r = P_d = 10^{-6}$, $\alpha = 0.15$dB/km, and $\lambda_m = 1$dB.}
\label{fig:Q_vs_p2_new}
\end{figure}

Before we develop a phenomenological model for $N_{\rm max}(p(2))$, let us get a feel for the dependence by extracting estimates of $N_{\rm max}(p(2))$ from the rate-loss plots shown in Fig.~\ref{fig:p_2_non_zero}. A good estimate can be obtained by assessing the value of $p(2)$ when an $N$-link concatenation becomes next to useless, one way to quantify which is when the maximum range for the $N$-link concatenation becomes less than $10$ km. Another way to quantify $N_{\rm max}$ would be to use the value of $p(2)$ for which the $N$-link concatenation's maximum range falls below the maximum range obtained with $N=1$ (that range threshold could be used as $590$ km for the parameters used in Fig.~\ref{fig:p_2_non_zero}, since the maximum range with $N=1$ is $600$ km).

In Fig.~\ref{fig:Q_vs_p2_new}(a) and (b), we plot the end-to-end QBER when a fixed overall range $L$ (of $590$ km, and $10$ km, respectively) is divided up into $1$, $2$, $4$ or $8$ elementary links. The color convention is the same as the one used for the secret key rate plots in Fig.~\ref{fig:p_2_non_zero}. The black horizontal lines correspond to $Q_{\rm th} = 0.1104$. The secret key rate goes to zero when the end to end QBER exceeds $Q_{\rm th}$. It is instructive to tally the $p(2)$ values where $Q_{n+1}$ crosses the $Q_{\rm th}$ line for $n=0, 1, 2, 3$, with the plots in Fig.~\ref{fig:p_2_non_zero}. The $p(2)$ value when the $8$-elementary-link chain's maximum range is $590$ km, is $0.0054$, and that when it is $10$ km is $0.0084$, both of which match well with the plots (c) and (d) of Fig.~\ref{fig:p_2_non_zero}. Similarly, the $p(2)$ value when the $4$-elementary-link chain's maximum range is $590$ km, is $0.0116$, and that when it is $10$ km is $0.0195$, which match well with plot (g) of Fig.~\ref{fig:p_2_non_zero}. Finally, the $p(2)$ value when the $2$-elementary-link chain's maximum range is $590$ km, is $0.0347$, and that when it is $10$ km is $0.0577$, which match well with plots (j), (k) and (l) of Fig.~\ref{fig:p_2_non_zero}. In the table in Fig.~\ref{fig:Q_vs_p2_new}(c), we record the values of $p(2)$, using the $10$ km estimate rule, corresponding to $N_{\max} = 8, 4$ and $2$. Our goal for the remainder of this section, will be to extract a phenomenological model for $N_{\max}(p(2))$---by quantifying how the QBER propagation law in Eq.~\eqref{eq:errorpropagation} must be modified when $p(2)>0$---that closely matches the estimates in Fig.~\ref{fig:Q_vs_p2_new}(c).

\begin{figure}
\centering
\includegraphics[width=\columnwidth]{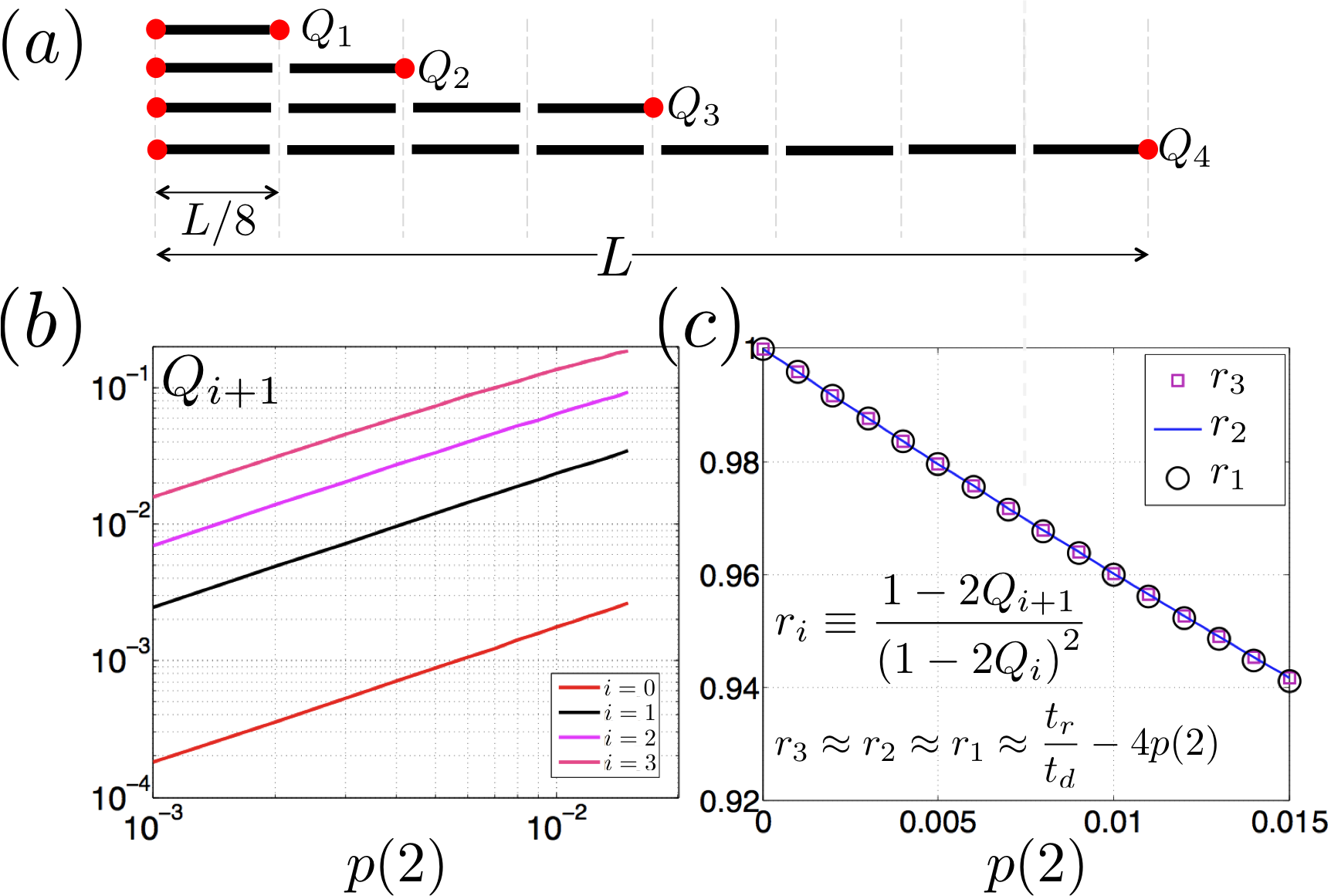}
\caption{(Color online) (a) Schematic showing the chain with $1$, $2$, $4$, and $8$ links. $Q_i$ is the QBER if Alice and Bob were to make measurements across a $2^i$-link chain. (b) $Q_{i+1}$ vs. two-pair-emission probability $p(2)$, for different numbers of swaps ($i=0, 1, 2, 3$) at a fixed distance of $L=50$ km (a short range is chosen to ensure that for all four cases the elementary-link quality is very good for the entire $p(2)$ range we consider, so that we cleanly capture the effect of $p(2)$ on the QBER). (c) Here we plot the ratio $r_i = (1-2Q_{i+1})/(1-2Q_i)^2$ as a function of $p(2)$, which shows that the ratio $r_i$ remains unchanged over $i=1, 2, 3$, hence suggesting that the QBER scaling law in Eq.~\eqref{eq:errorpropagation} holds even when $p(2) > 0$. For all plots, we assume, $\eta_e=\eta_r=\eta_d = 0.9$, $P_e = P_r = P_d = 10^{-6}$, $M=1000$, $\alpha = 0.15$dB/km, $\lambda_m = 1$dB, and $T_q = 50$ns.}
\label{fig:Q_vs_p2}
\end{figure}
{\em QBER propagation}---In Fig.~\ref{fig:Q_vs_p2}(a), we depict our $L$-km-range, $N=2^n$ elementary-link construction, for $n=3$. The Alice-to-Bob range $L$ is divided up into $N = 2^n$ elementary links, and $Q_i$ is defined as the error probability {\em if} Alice and Bob were to measure the state $\rho_i$ (which is formed after successfully connecting $2^{i-1}$ elementary links, each of length $L/N$), $1 \le i \le n+1$. In Fig.~\ref{fig:Q_vs_p2}(b), we plot $Q_i$ as a function of $p(2)$, when $p(1) = 0.9$ is held fixed, with $p(0) = 1-p(1)-p(2)$, for $N=2^n$, with $n=3$. At each value of $i \in \left\{0, 1, 2, 3\right\}$, the respective QBER $Q_{i+1}$ seems to grow almost linearly with $p(2)$ when $p(2)$ is small, for chosen system parameters as mentioned in the caption of Fig.~\ref{fig:Q_vs_p2}. In Fig.~\ref{fig:Q_vs_p2}(c), we plot the ratio, $C(p(2)) = (1-2Q_{i+1})/(1-2Q_i)^2$ for $i = 1, 2, 3$, as a function of $p(2)$. For the ideal source ($p(2)=0$), we proved that the QBER ratio $C(p(2)) = t_r/t_d$, which is independent of $i$; see Eq.~\eqref{eq:errorpropagation}. For the aforesaid loss and noise parameters, $t_r/t_d = 1 - \epsilon$, with $\epsilon = 1.39 \times 10^{-5}$. We see here numerically, that $C(p(2))$ is independent of $i$, {\em even for an imperfect source}, for any value of $p(2) \in [0, 0.055]$. The ratio has a good fit to the line, $C \approx (t_r/t_d) - 4p(2)$ for the above range of $p(2)$. The $p(2)$-dependence of $C$ deviates from linear as $p(2)$ becomes higher. This is quite interesting, as this gives us a way to predict the end-to-end QBER on long repeater chains by making a physical measurement on one noisy elementary link, if similar devices are used to construct each elementary link.

\begin{figure}
\centering
\includegraphics[width=\columnwidth]{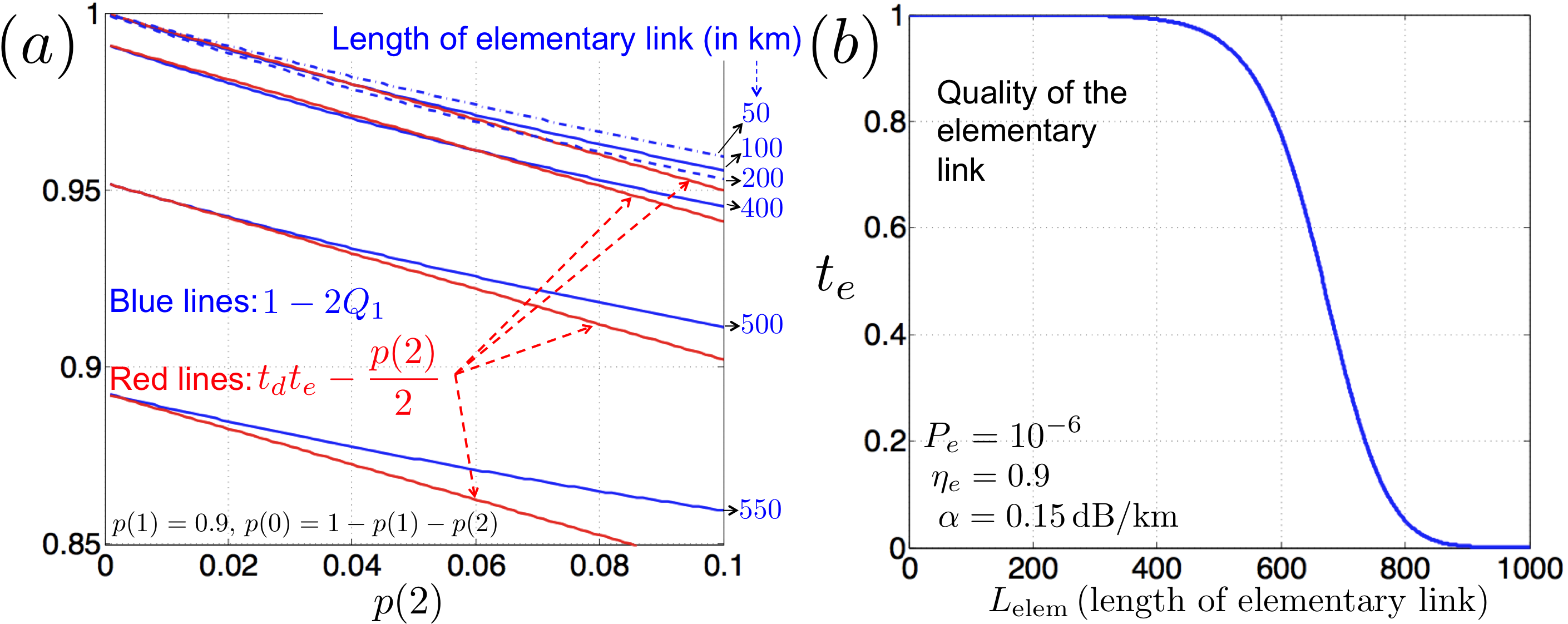}
\caption{(Color online) (a) Plot of $1- 2Q_1$, where $Q_1$ is the QBER of one elementary link (of range $L_{\rm elem}= L/N$, chosen in the range $50$ km to $550$ km), as a function of $p(2)$. It is seen that, $1 - 2Q_1 \gtrsim t_dt_e - \frac12{p(2)}$. (b) Plot of $t_e$, the `quality' of one elementary link, as a function of $L_{\rm elem}$, for $P_e = 10^{-6}$, $\eta_e = 0.9$, and $\alpha = 0.15$ dB/km. For $L_{\rm elem} < 400$ km, $t_e$ is seen to remain close to $1$.}
\label{fig:elemlinkquality}
\end{figure}
{\em QBER of one elementary link}---In Fig.~\ref{fig:elemlinkquality}(a), we plot $1- 2Q_1$, with $Q_1$ the QBER of one elementary link (of range $L_{\rm elem}= L/N$, chosen in the range $50$ km to $550$ km), as a function of $p(2)$. It is seen that,
\begin{equation}
1 - 2Q_1 \gtrsim t_dt_e - \frac12{p(2)}.\label{eq:Q1}
\end{equation}
This linear approximation seems good for $L_{\rm elem} \lesssim 400$ km, and for $p(2) < 0.02$. We next put this together with the linear approximation of the constant in the QBER scaling law, i.e.,
\begin{equation}
1 - 2Q_{i+1} \gtrsim \left(\frac{t_r}{t_d} - 4p(2)\right)(1-2Q_i)^2, \,i \ge 1.\label{eq:Q2}
\end{equation}
Simplification of the recursion in Eq.~\eqref{eq:Q2} yields,
\begin{eqnarray}
1 - 2Q_i &\ge& \left(\frac{t_r}{t_d} - 4p(2)\right)^{2^0 + 2^1 + \ldots + 2^{i-2}}(1-2Q_1)^{2^{i-1}} \nonumber \\
&=&  \left(\frac{t_r}{t_d} - 4p(2)\right)^{2^{i-1} - 1}(1-2Q_1)^{2^{i-1}},
\end{eqnarray}
which combined with Eq.~\eqref{eq:Q1} yields
\begin{equation}
1-2Q_i \gtrsim \left(\frac{t_r}{t_d} - 4p(2)\right)^{2^{i-1} - 1}\,\left(t_dt_e - \frac{1}{2}p(2)\right)^{2^{i-1}}.
\end{equation}
Taking logarithms, rearranging the terms, and noting that each of the three terms $\log(1-2Q_i)$, $\log(t_r/t_d - 4p(2))$, and $\log(t_dt_e - \frac12 p(2))$ are negative, we get the following:
\begin{equation}
2^{i-1} \gtrsim \frac{\left| \log(1-2Q_i) + \log(t_r/t_d - 4p(2))\right|}{\left|\log(t_dt_e - \frac12{p(2)}) + \log(t_r/t_d - 4p(2))\right|}.\label{eq:Nmaxfullexp}
\end{equation}
Note now that $Q_i \equiv Q(N)$ is the QBER if Alice and Bob were to make an end-to-end measurement on $N = 2^{i-1}$ elementary links (see Fig.~\ref{fig:Q_vs_p2}(a)). Hence the condition on $2^{i-1}$ to be the {\em maximum} total number of elementary links (i.e., $N = N_{\max}$) for which a barely non-zero key rate can be obtained, is that $Q(N) = Q_{\rm th}$.

{\em Phenomenological model for $N_{\max}$}---Substituting $\log(1-2Q_i) = \log(1-2Q_{\rm th}) \approx -0.25$, $\log(1-x) \approx -x-x^2/2$, and $t_r = t_d = t_e = 1$ (in order to capture the $N_{\rm max}(p(2))$ dependence, and do so in the low-noise regime of the elementary links) in Eq.~\eqref{eq:Nmaxfullexp}, and ignoring the ${\mathcal O}(p(2)^2)$ terms, we obtain the following approximate lower estimate to $N_{\rm max}$,
\begin{equation}
N_{\rm max} \gtrsim (8/9) + \frac{1/18}{p(2)},
\label{eq:Nmaxp2}
\end{equation}
which is roughly a shifted inverse-proportional dependence in $p(2)$. The above interpretation of $N_{\rm max}$ is that it is the maximum number of length $L_{\rm elem}$ elementary links that can be connected before the concatenation becomes useless for QKD (while using $N < N_{\rm max}$ links is capable of attaining the $p(2)=0$ rate-distance function $R_N(L)$ derived in Section~\ref{sec:analysis}). The `quality' of the elementary link is captured by the parameter $t_e$---defined for the $p(2)=0$ analysis in Section~\ref{sec:analysis}---which is $1$ when the dark click probability of the detectors at the center of the elementary link, $P_e = 0$. In Fig.~\ref{fig:elemlinkquality}(b), we plot $t_e$ as a function of the length of the elementary link $L_{\rm elem} \equiv L/N$, for $P_e = 10^{-6}$, $\eta_e = 0.9$, and $\alpha = 0.15$ dB/km. For $L_{\rm elem} < 400$ km, $t_e$ is seen to remain close to $1$. This justifies substituting $t_e = 1$ in order to arrive at Eq.~\eqref{eq:Nmaxp2}. The table in Fig.~\ref{fig:Q_vs_p2_new}(c) shows that the $N_{\rm max}(p(2))$ lower estimate we obtained indeed matches pretty well with the exact values obtained numerically shown in Figs.~\ref{fig:Q_vs_p2_new}(a--b). We must note here, that we do not consider the effect of the number of modes $M$ on $N_{\rm max}$ (which we hold fixed for the above development).

\section{Conclusions}\label{sec:conclusions}
Long-distance entanglement distribution at high rates is of paramount importance to many quantum communication protocols, the realization of which requires building a network of quantum repeaters. Several quantum repeater protocols have been proposed~\cite{Bri98, Dua01, San11, Jia09, Sin13A, Lo13}, all of which use some source of entanglement, some form of quantum memories, and linear-optics-based Bell-state measurements. We analyzed the architecture proposed in~\cite{Sin13A}, which is a repeater protocol that has a superior classical communication overhead, and does not rely on purification of noisy shared entangled pairs~\cite{Kro14}. We believe that our analysis technique would carry over to other repeater architectures in a straightforward manner.

We exactly solved for the quantum state after connecting a given number of elementary links in a concatenated quantum-repeater chain that uses frequency multiplexing to create two-qubit four-photon elementary link states, and heralded linear-optic Bell-state measurements (BSM) at a pre-determined frequency across two qubit memories at repeater nodes. We exploited the fact that if we start with an ideal single-pair entanglement source, the post-selected state after a successful BSM remains in a subspace spanned by only single photon terms, and we recursively evaluated the end-to-end entangled state using a POVM to model lossy-noisy single-photon detectors. This calculation required us to exactly solve a variant of the {\em logistic map} from chaos theory. Using our expression for the quantum state, we determined quantities such as the success probability of entanglement swapping at any given swap level, the error rate of the raw bits obtained by Alice and Bob in a QKD application if they were to measure this state in the same bases, and the sifting probability. One can find any other quantity of interest from the quantum state, such as the entanglement of formation or the fidelity with a maximally entangled state (see Appendix~\ref{app:QBER} for the exact expression of fidelity of the $N$-elementary link end-to-end state). Our analysis took into account all major imperfections of the detectors (such as sub-unity detection efficiencies, and dark click probabilities) and the channel (such as transmissivity and thermal noise, where the latter can be included into an effective dark-click probability term). We also evaluated an exact scaling law for how the quantum bit error rate (QBER) evolves from one swap level to the next, which is of great practical importance since it gives us a way to predict the QBER on long repeater chains by making a physical measurement on one noisy elementary link.

We evaluated the rate-vs.-loss envelope attained by this repeater-chain architecture, and showed that the secret-key rate achieved can be expressed as $R = A\eta^\xi$, where $\eta$ is the overall Alice-to-Bob channel transmittance, and $A$ and $\xi < 1$ are constants that depend upon various loss and noise parameters of the system. This in turn proved that the repeater chain's performance beats the TGW bound, a fundamental rate-loss upper bound that no QKD protocol can exceed without the use of quantum repeaters~\cite{Tak13}, which imposes a linear rate-transmittance decay (i.e., $\xi =1$). This, to our knowledge, is one of the first rigorous proofs of the efficacy of any quantum repeater protocol.

We then extended our theoretical analysis to the case when the entangled photon pair sources have a non-zero two-pair emission probability, $p(2)$. For this, we used an efficient numerical model we developed for simulating bosonic states, linear-optic unitaries, and noisy measurements. We found that when $p(2) > 0$, the rate-distance tradeoff plots---with $N$ elementary links dividing up the entire range $L$ km---are almost unaffected (i.e., remain almost at their $p(2)=0$ levels at any range $L$), for all $N$ up to below a maximum value $N_{\rm max}$, where $N_{\rm max}(p(2))$ decreases as $p(2)$ is increased. If $N_{\rm max}(p(2))$ or more elementary links are used, the key rate is worse at all range $L$ compared to when fewer elementary links are used. Finally, we developed a phenomenological model for $N_{\rm max}(p(2))$ by an empirical extension of the aforesaid QBER scaling law for the $p(2)>0$ case. One of the most commonly employed optical entanglement sources uses spontaneous parametric downconversion (SPDC) devices heralded by single photon detectors~\cite{Sli03}. SPDC sources have a high enough non-zero $p(2)$ to render them ineffective as sources for the repeater protocol as described in this paper. In a subsequent paper~\cite{Kro15}, we show how photon number resolving detectors can be employed to obtain an improved sifting performance by post-selecting out erroneous multi-photon events stemming from non-zero $p(2)$, and thereby making it possible to retrieve the good rate-vs.-distance scaling. 

One can in principle replace the linear-optic entanglement swapping scheme with more advanced schemes with improved heralding efficiencies, such as the one proposed in Ref.~\cite{Gri11} that injects entangled states into a beamsplitter network and heralds the total number of clicks from an array of photon-number-resolving detectors, one that uses inline squeezers to beat the $50\%$-efficiency limit of a linear-optic BSM~\cite{Zai13}, and another proposal that can attain $75\%$ or higher heralding efficiencies via linear-optics and injection of (un-entangled) single-photons~\cite{Ewe14}. Our theoretical technique can be readily used to analyze the repeater-chain when the BSMs are replaced by one of the aforesaid schemes. At each swap stage, after the post-selection by the BSM, the projected shared state will still lie in the span of the $4$-mode $2$-qubit `dual-rail' basis, but there will be two extra coefficients to track, since the advanced BSMs can identify all four Bell states (as opposed to only two by the linear-optic scheme~\cite{Lut99}). It is quite likely that the final expression for $Q_i$, and the error-propagation law will still depend upon $t_d$, $t_r$, and $t_e$, where the latter two are the same functions of the fractional probability transfer to classical correlations at each swap stage (which should be smaller compared to when the linear-optic BSM is used). Finally, our numerical model allows us to evaluate these enhanced schemes as well, and also introduce other non-idealities such as finite memory times at the repeaters, non-linearities in the fiber and memories, and temporal non-idealities of single photon detectors such as timing jitter and after-pulsing probabilities. The analysis of quantum repeater protocols that use these advanced BSM schemes, a possible extension where multiplexing extends across elementary links (i.e. using more than one connection between elementary links), and protocols that may use quantum purification at intermediate stages, are left for future work. Furthermore, we hope that the compact rate-loss scaling results we developed in this paper for a linear repeater chain will help seed future network theoretic analyses, for instance optimal rate regions for multi-flow routing, traffic scheduling, and resource allocation, in a quantum network with more complex topologies. Finally, we expect our work to incite similar rate-loss analysis of other quantum repeater protocols, which will enable quantitative resource-performance tradeoff-studies and meaningful comparisons of different protocols.

\begin{acknowledgments}
{The authors would like to thank Khabat Heshami, Gregory Kanter and Yuping Huang for useful discussions. SG thanks Rodney Van{ }Meter and Mohsen Razavi for detailed feedback on an earlier version of this manuscript, and thanks Masahiro Takeoka and Donald Towsley for useful discussions. This paper is based on research funded by the DARPA Quiness program subaward contract number SP0020412-PROJ0005188, under prime contract number W31P4Q-13-1-0004. WT, a senior fellow of the Canadian Institute for Advanced Research (CIFAR), also acknowledges support from Alberta Innovates Technology Futures (AITF). The views and conclusions contained in this document are those of the authors and should not be interpreted as representing the official policies, either expressly or implied, of the Defense Advanced Research Projects Agency, or the U.S. Government.}
\end{acknowledgments}

\appendix

\section{Proof of Proposition \ref{prop:connections}: Quantum state of the elementary link, and entangled state propagation through a sequence of swap stages}\label{app:elem}

\subsection{The elementary link}\label{app:elemlink}
We first prove Proposition~\ref{prop:connections} for the case $i=1$, and derive the post-selected quantum state of the elementary link. Let us first consider how we should model non-ideal photodetectors. Ideally we would like to say that each of the four detectors required for the BSM individually measures a Hermitian operator with eigen-projectors $\{\Pi_0,\Pi_1, \Pi_2\}$, the $\Pi_n=|n\rangle\langle n|$ signifying the presence of $n$ photons. Next we note that we are allowed to limit ourselves to a three-dimensional subspace of the Fock space because we know we will never have more than two photons at a detection site (since we limit the theoretical part of analysis to the case when the sources have $p(2)=0$ and assume that any thermal noise in the channel is negligible at typical optical frequencies). The detectors are assumed to have a sub-unity detection efficiency $\eta_e$---which may be thought of as arising from a beamsplitter with transmissivity ${\eta_e}$ just in front of an ideal detector---and independently there may also be a probability $P_e$ for the detector to trigger in the absence of a photon. This means the ``no click'' and ``click'' events in the individual detectors really correspond to a two-outcome POVM $\{F_0,F_1\}$, with
\begin{eqnarray}
F_0 &=& (1-P_e)\Pi_0 + (1-A_e)\Pi_1 + (1-B_e)\Pi_2
\\
F_1 &=& P_e\Pi_0 + A_e\Pi_1 + B_e\Pi_2\;.
\end{eqnarray}
where we take
\begin{eqnarray}
A_e &=& 1-(1-P_e)(1-\eta)
\\
B_e &=& 1-(1-P_e)(1-\eta)^2\;.
\end{eqnarray}
The way to understand $F_0$, the ``no click'' signal for instance, is this:  If there are no actual photons present, one will get this outcome with probability $1-P_e$, the probability for no false alarm at the detector.  On the other hand, if there is a single photon present both it must disappear and there still be no false alarm; hence a coefficient $(1-P_e)(1-\eta)$ in front of $\Pi_1$. Finally, for the case that two photons are present, both of them must be lost and yet no false alarm must appear; hence a coefficient of $(1-P_e)(1-\eta)^2$.

We next note that we may incorporate the channel transmittance $\lambda$ (corresponding to propagation loss of each of the halves of the Bell pairs from two ends of the elementary link) directly into the detection efficiency $\eta_e$, by defining an effective detection efficiency $\eta_e\lambda$ while assuming the channel is lossless, rather than accounting for the channel loss in our description of the quantum states arriving at them.  One can see this through a simple bosonic mode-operator analysis including two stages of loss, but the intuition should be clear. Consequently, at the center of an elementary link we can assume the state it will attempt to link is a clean $|M^+\rangle|M^+\rangle$, while the four detectors in the BSM are working at efficiency
\begin{equation}
\eta=\eta_e\lambda\;.
\end{equation}
This greatly simplifies the analysis by not having to treat the states to be linked as mixed states.

For the purposes of the derivations in this subsection, let us label the four spatial modes involved in an elementary link by $a$, $b$, $c$, and $d$, so that the initial quantum state is more explicitly $|M_{ab}^+\rangle|M_{cd}^+\rangle$.  The BSM will be applied to modes $b$ and $c$.  What this entails is that the modes first impinge on a 50-50 beamsplitter, which enacts a mode transformation
\begin{equation}
b^\dagger_j \; \longrightarrow \; \sqrt{\frac{1}{2}}\,\Big(b^\dagger_j+c^\dagger_j\Big) \quad \mbox{and} \quad
c^\dagger_j \; \longrightarrow \; \sqrt{\frac{1}{2}}\,\Big(b^\dagger_j-c^\dagger_j\Big)\;.
\end{equation}
The consequence of this is that the state presented to the photo detectors is a massively entangled one:
\begin{eqnarray}
|\,\mbox{swap}\rangle &=& \frac{1}{4}\,\left[|10,11,00,01\rangle - |10,01,10,01\rangle \right. \nonumber \\
&&+ \sqrt{2}|10,02,00,10\rangle+|10,10,01,01\rangle
\nonumber\\
&& - |10,00,11,01\rangle - \sqrt{2}|10,00,02,10\rangle \nonumber\\
&&+ \sqrt{2}|01,20,00,01\rangle + |01,11,00,10\rangle
\nonumber\\
&& - |01,10,01,10\rangle - \sqrt{2}|01,00,20,01\rangle \nonumber\\
&&\left. + |01,01,10,10\rangle - |01,00,11,10\rangle\right]
%\nonumber\\
\label{swapper}
\end{eqnarray}

Ideally then, if one were to obtain a 1-2 coincidence or a 3-4 coincidence in the detectors at the four dual-rail modes, a successful entanglement swap would be declared and a new state $|M^+_{ad}\rangle$ would be ascribed to the photons in quantum memory.  However with noisy detectors, one should use L\"uders' rule for the POVM above to get the new state.  For instance, suppose we were to detect a 1-2 coincidence in the detectors.
Then, this is signified by the POVM element
\begin{eqnarray}
&&F_1\otimes F_1 \otimes F_0\otimes F_0 \nonumber\\
&=& P_e^2(1-P_e^2)^2\Pi_0\otimes\Pi_0\otimes\Pi_0\otimes\Pi_0 \nonumber\\
&+& P_e^2(1-P_e^2)(1-A_e)\Pi_0\otimes\Pi_0\otimes\Pi_0\otimes\Pi_1+\ldots
\end{eqnarray}
and the new state for the $a$-$d$ system will be
\begin{widetext}
\be
\rho^\prime_{ad}=\frac{1}{{\rm Prob}(F_1\otimes F_1 \otimes F_0\otimes F_0)}{\rm tr}_{bc}\!\left( \sqrt{F_1\otimes F_1 \otimes F_0\otimes F_0}\, |\mbox{swap}\rangle\langle\mbox{swap}| \, \sqrt{F_1\otimes F_1 \otimes F_0\otimes F_0}\right).
\ee
\end{widetext}

From here on out is just a question of brute-force calculation. At the end of it, one finds:
\begin{eqnarray}
\rho_{ad}^\prime &=&
\frac{1}{8s_1}\left\{
\Big[A_e^2 (1-P_e)^2 + P_e^2 (1-A_e)^2\Big]\,|M^+_{ad}\rangle\langle M^+_{ad}| \right. \nonumber\\
&&+ \; 2A_e P_e(1-A_e)(1-P_e)|M^-_{ad}\rangle\langle M^-_{ad}|
\nonumber\\
&&+\;
P_e(1-P_e)\Big[P_e(1-B_e)+B_e(1-P_e)\Big] \times \nonumber\\
&&\left. \Big(|01,01\rangle\langle01,01|+|10,10\rangle\langle10,10|\Big)\right\},
\end{eqnarray}
where, the success probability to herald an elementary link $\rho_1$, $P_{s0} = {{\rm Prob}(F_1\otimes F_1 \otimes F_0\otimes F_0)} = 4s_1$, where
\begin{eqnarray}
s_1 &=& \frac{1}{8} \left[ (A_e + P_e - 2 A_e P_e)^2 \right.\nonumber\\
&&\left. + P_e(1-P_e)(B_e + P_e - 2 B_e P_e)\right]\;.
\end{eqnarray}
Thus one has mostly the swap expected.  But with some probability one gets an unexpected swap, and with some probability an induced classical correlation between the photons in the memory. By symmetry one has the same result for a 3-4 coincidence, and for 1-4 and 2-3 coincidences, one just interchanges the roles of $|M^+_{ad}\rangle$ and $|M^-_{ad}\rangle$ in this expression. We therefore have the state of an elementary link given by:
\begin{eqnarray}
\rho_1 &=& \frac{1}{s_1}\left[a_1|M^+\rangle\langle M^+| + b_1|M^-\rangle\langle M^-| + c_1|\psi_0\rangle\langle\psi_0| \right.\nonumber\\
&&\left. +\; d_1|\psi_1\rangle\langle\psi_1| + d_1|\psi_2\rangle\langle\psi_2| + c_1|\psi_3\rangle\langle\psi_3|\right],
\end{eqnarray}
where $|\psi_0\rangle = |01,01\rangle$, $|\psi_1\rangle = |01,10\rangle$, $|\psi_2\rangle = |10,01\rangle$, $|\psi_3\rangle = |10,10\rangle$, $|M^{\pm}\rangle = \left[|\psi_2\rangle \pm |\psi_1\rangle\right]/\sqrt{2}$, $s_1 = a_1 + b_1 + 2(c_1 + d_1)$ is a normalization constant, and the coefficients $a_1$, $b_1$, $c_1$, $d_1$ are given by:
\begin{eqnarray}
a_1 &\equiv& a_e = \frac{1}{8}\left[P_e^2(1-A_e)^2 + A_e^2(1-P_e)^2\right],\nonumber \label{eq:ae}\\
b_1 &\equiv& b_e = \frac{1}{8}\left[2A_eP_e(1-A_e)(1-P_e)\right], \label{eq:be}\nonumber\\
c_1 &\equiv& c_e = \frac{1}{8}P_e(1-P_e)\left[P_e(1-B_e)+B_e(1-P_e)\right], \label{eq:ce}\nonumber\\
d_1 &\equiv& d_e = 0,\nonumber
\end{eqnarray}
where $A_e = \eta_e\lambda + P_e(1 - \eta_e\lambda)$ and $B_e = 1-(1-P_e)(1 - \eta_e\lambda)^2$.

\subsection{Connections through swap stages at the quantum repeater nodes}
Next we consider the case $i\geq 2$. The proof proceeds as follows. We first realize, by term-by-term evaluation of connecting two copies of $\rho_1$, that the state $\rho_i$ never goes outside the span of $|\psi_0\rangle, |\psi_1\rangle, |\psi_2\rangle, |\psi_3\rangle$. It is convenient to express the state $\rho_i$ as:
\begin{eqnarray}
\rho_i &=&\frac{1}{s_i}\left[r_1^{(i)}|M^+\rangle\langle M^+| + r_2^{(i)}|M^-\rangle\langle M^-| + r_3^{(i)}|\psi_0\rangle\langle\psi_0| \right.\nonumber\\
&&\left. +\; r_4^{(i)}|\psi_1\rangle\langle\psi_1| + r_5^{(i)}|\psi_2\rangle\langle\psi_2| + r_6^{(i)}|\psi_3\rangle\langle\psi_3|\right],
\end{eqnarray}
where $s_i = \sum_{l=1}^6r_l^{(i)}$. Then, we realize that each subsequent connection evolves the state as,
\begin{equation}
r_l^{(i+1)} = \sum_{j=1}^6 \sum_{k=1}^6 C_{j,k,l}r_j^{(i)}r_k^{(i)},
\end{equation}
with the matrix $C$ given by (each term of which is calculated by brute-force algebra):
\begin{eqnarray}
C(1,1,:) &=& [a, b, c, 0, 0, c] \nonumber \\
C(1,2,:) &=& [b, a, c, 0, 0, c]  \nonumber \\
C(1,3,:) &=& [0, 0, a+b, 0, 2c, 0]  \nonumber \\
C(1,4,:) &=& [0, 0, 0, a+b, 0, 2c]  \nonumber \\
C(1,5,:) &=& [0, 0, 2c, 0, a+b, 0]  \nonumber \\
C(1,6,:) &=& [0, 0, 0, 2c, 0, a+b],  \nonumber
\end{eqnarray}
\begin{eqnarray}
C(2,1,:) &=& [a, b, c, 0, 0, c]  \nonumber \\
C(2,2,:) &=& [b, a, c, 0, 0, c]  \nonumber \\
C(2,3,:) &=& [0, 0, a+b, 0, 2c, 0]  \nonumber \\
C(2,4,:) &=& [0, 0, 0, a+b, 0, 2c]  \nonumber \\
C(2,5,:) &=& [0, 0, 2c, 0, a+b, 0]  \nonumber \\
C(2,6,:) &=& [0, 0, 0, 2c, 0, a+b],  \nonumber
\end{eqnarray}
\begin{eqnarray}
C(3,1,:) &=& [0, 0, a+b, 2c, 0, 0]  \nonumber \\
C(3,2,:) &=& [0, 0, a+b, 2c, 0, 0]  \nonumber \\
C(3,3,:) &=& [0, 0, 4c, 0, 0, 0]  \nonumber \\
C(3,4,:) &=& [0, 0, 0, 4c, 0, 0]  \nonumber \\
C(3,5,:) &=& [0, 0, 2(a+b), 0, 0, 0]  \nonumber \\
C(3,6,:) &=& [0, 0, 0, 2(a+b), 0, 0],  \nonumber
\end{eqnarray}
\begin{eqnarray}
C(4,1,:) &=& [0, 0, 2c, a+b, 0, 0]  \nonumber \\
C(4,2,:) &=& [0, 0, 2c, a+b, 0, 0]  \nonumber \\
C(4,3,:) &=& [0, 0, 2(a+b), 0, 0, 0]  \nonumber \\
C(4,4,:) &=& [0, 0, 0, 2(a+b), 0, 0]  \nonumber \\
C(4,5,:) &=& [0, 0, 4c, 0, 0, 0]  \nonumber \\
C(4,6,:) &=& [0, 0, 0, 4c, 0, 0],  \nonumber
\end{eqnarray}
\begin{eqnarray}
C(5,1,:) &=& [0, 0, 0, 0, a+b, 2c]  \nonumber \\
C(5,2,:) &=& [0, 0, 0, 0, a+b, 2c]  \nonumber \\
C(5,3,:) &=& [0, 0, 0, 0, 4c, 0]  \nonumber \\
C(5,4,:) &=& [0, 0, 0, 0, 0, 4c]  \nonumber \\
C(5,5,:) &=& [0, 0, 0, 0, 2(a+b), 0]  \nonumber \\
C(5,6,:) &=& [0, 0, 0, 0, 0, 2(a+b)], \nonumber
\end{eqnarray}
\begin{eqnarray}
C(6,1,:) &=& [0, 0, 0, 0, 2c, a+b]  \nonumber \\
C(6,2,:) &=& [0, 0, 0, 0, 2c, a+b]  \nonumber \\
C(6,3,:) &=& [0, 0, 0, 0, 2(a+b), 0]  \nonumber \\
C(6,4,:) &=& [0, 0, 0, 0, 0, 2(a+b)]  \nonumber \\
C(6,5,:) &=& [0, 0, 0, 0, 4c, 0]  \nonumber \\
C(6,6,:) &=& [0, 0, 0, 0, 0, 4c],
\end{eqnarray}
where the ``$\colon$" sign indicates all entries $C(j,k,l)$ for $1 \le l \le 6$. The rest is just writing out $r_l^{(i+1)}$ explicitly, and realizing that,
\begin{eqnarray}
r_3^{(i)} &=& r_6^{(i)}, \,{\text{and}}\\
r_4^{(i)} &=& r_5^{(i)},
\end{eqnarray}
and hence the fact that we can rename the coefficients as: $r_1^{(i)} = a_i, r_2^{(i)} = b_i, r_3^{(i)} = r_6^{(i)} = c_i$, and $r_4^{(i)} = r_5^{(i)} = d_i$.

\section{Evaluating the success probabilities}\label{app:success}

It is easy to realize from the derivation of the states $\rho_i$ that the success probability (to connect two copies of $\rho_{i-1}$ to obtain one copy of $\rho_i$) is simply given by $P_s(i) = 4s_i$, for $i \ge 2$. The probability an elementary link is successfully created is $P_s(1) = 1 - (1 - P_{s0})^M$, where $P_{s0} = 4s_1$ is the probability of successful creation of an elementary link $\rho_1$ in one of the $M$ frequencies at the center of the elementary link, where $s_1 = a_e+b_e+2c_e$. It is simple now to calculate the success probabilities $P_s(i)$ by proving that $s_i = s$, $\forall i \ge 2$. We thus have the following proposition.
\begin{proposition}
The success probability of connecting two copies of $\rho_{i-1}$ to produce a usable copy of $\rho_i$, $P_s(i) = 4s_i$, where
\begin{equation}
s_i = a + b + 2c \triangleq s, \, 2 \le i \le n+1.
\end{equation}
\end{proposition}

\begin{proof}
Denoting $x_i = a_i + b_i + c_i +d_i$, and $y_i = c_i + d_i$, using Eqs.~\eqref{eq:aip1},~\eqref{eq:bip1},~\eqref{eq:cip1},~\eqref{eq:dip1}, we have,
\begin{eqnarray}
x_{i+1} &=& \frac{1}{s_i^2}\left[(a+b+c)(x_i^2 + y_i^2) + 2cx_iy_i\right],\label{eq:x}\\
y_{i+1} &=& \frac{1}{s_i^2}\left[c(x_i-y_i)^2 + 2(a+b+2c)x_iy_i\right],\label{eq:y}
\end{eqnarray}
with $s_i = x_i + y_i$ by definition. It is easy to now see that $x_{i+1} + y_{i+1} = a+b+2c \equiv s$, for all $i \in \left\{2, 3, \ldots, n+1\right\}$. Note that $P_s(1) = 1 - (1 - 4s_1)^M$, with $s_1 = a_e + b_e + 2c_e$ for the elementary link.
\end{proof}

\section{Evaluating the sift probability}\label{app:sift}

In this Appendix, we derive $P_1$, the probability that Alice and Bob get a successful `sift', i.e., they decide to use their click outcomes for further processing to extract a key when they measure their halves of the shared entangled state $\rho_{n+1}$ (given $N = 2^n$ elementary links have been connected successfully).

Let us first assume Alice and Bob share the state $\rho_i$, and they make a measurement (in the same basis). We proceed as follows.
\begin{proposition}\label{prop:P1}
The sift probability $P_1$ is the probability that Alice and Bob both get clicks on at least one of each of their detectors (i.e., neither gets a no-click event on both detectors). Regardless of the value of $i$,
\begin{equation}
P_1 = (q_1 + q_2 + q_3)^2,
\end{equation}
where $q_1 = (1-P_d)A_d$, $q_2 = (1-A_d)P_d$, $q_3 = P_dA_d$, with $A_d = \eta_d + (1 - \eta_d)P_d$, functions of the detection efficiency ($\eta_d$) and dark-click probability ($P_d$) of each of the four single-photon detectors involved (two of Alice's and two of Bob's).
\end{proposition}

\begin{proof}
This can be shown rigorously by simply evaluating $P_1 = {\rm Tr}[\rho_i(M_{0101} + M_{0110} + M_{1001} + M_{1010} + M_{1101} + M_{1110} + M_{0111} + M_{1011} + M_{1111})]$, and $M_{ijkl} \equiv F_i \otimes F_j \otimes F_k \otimes F_l$, where the POVM elements of a lossy-noisy single-photon detector, $F_0$ and $F_1$ are defined above, using the expression of $\rho_i$ in Eq.~\eqref{eq:stateexpression_app}. Here we will sketch a more intuitive proof. Note that $\rho_i \in {\rm span}(|\psi_0\rangle, |\psi_1\rangle, |\psi_2\rangle, |\psi_3\rangle)$, with $|\psi_0\rangle = |01,01\rangle$, $|\psi_1\rangle = |01,10\rangle$, $|\psi_2\rangle = |10,01\rangle$, $|\psi_3\rangle = |10,10\rangle$, since $|M^{\pm}\rangle = \left[|\psi_2\rangle \pm |\psi_1\rangle\right]/\sqrt{2}$. Therefore, Alice's and Bob's reduced density operators always have exactly one photon in one of two modes. Let us define $q_1 \triangleq P[{\rm no flip}]$ to be the probability that a $|01\rangle$ state is detected as ``$(0, 1)$" by the lossy-noisy detector, where $(0, 1)$ stands for (no-click, click). Clearly, $q_1$ is also the probability that $|10\rangle$ is detected as ``$(1, 0)$". In order for ``no flip" to happen, no dark click should appear in the mode in the vacuum state (this happens with probability $1 - P_d$), and that the photon in the other mode should either be detected by the lossy detector (happens with probability $\eta_d$, in which case it does not matter whether or not a dark click appears), or the photon does not get detected, {\em and} a dark click appears (which happens with probability $(1-\eta_d)P_d$). Therefore, $q_1 = (1-P_d)A_d$, with $A_d = \eta_d + (1-\eta) P_d$. Similarly, we define $q_2 \triangleq P[{\rm flip}]$ to be the probability that $|01\rangle$ is detected as ``$(1, 0)$" (or, $|10\rangle$ is detected as ``$(0, 1)$"). For a ``flip" event to happen, a dark click should appear in the vacuum mode (probability $P_d$), and the photon containing mode should not be detected {\em and} a dark click must not appear (happens with probability, $(1-\eta_d)(1-P_d)$). Therefore, $q_2 = (1-\eta_d)(1-P_d)P_d = (1-A_d)P_d$. Finally, define $q_3$ to the probability that the ``$(1, 1)$" detection is obtained (either for a $|10\rangle$ or a $|01\rangle$ input). This is given by the probability that a dark click appears in the vacuum mode ($P_d$) and the probability that the single photon generates a click, i.e., $\eta_d + (1-\eta_d)P_d = A_d$. Therefore, $q_3 = P_dA_d$. Clearly, $q_1 + q_2 + q_3$ need not add up to $1$ in general, since one of two detectors may output the ``$(0, 0)$" outcome, which is when Alice and Bob discard the measurement---a failed sift event. Therefore $(q_1+q_2+q_3)^2$ is the probability that Alice and Bob obtain a {\em usable} detection outcome, i.e., both of them collectively obtain one of the nine detection outcomes: $(0, 1; 0, 1), (0, 1; 1, 0), (1, 0; 0, 1), (1, 0; 1, 0), (0, 1; 1, 1)$, $(1, 0; 1, 1), (1, 1; 0, 1), (1, 1; 1, 0), (1, 1; 1, 1)$. This is true regardless of the actual fraction of $|10\rangle$ and $|01\rangle$ in Alice's and Bob's states. Hence, $P_1 = (q_1+q_2+q_3)^2$.
\end{proof}

\section{The QBER and secret key rate}\label{app:QBER}

In this Appendix, we will evaluate the explicit formula for $Q_i$, the quantum bit-error rate (QBER), which is the probability that Alice and Bob obtain a mismatched raw key bit, despite the fact that they make measurements in the same bases on a successfully-created copy of $\rho_i$, {\em and} that they both get exactly single-clicks (on the two modes of their respective qubits). The first step in doing so is to solve for the quantum state $\rho_i$ more explicitly than what the recursions in Proposition~\ref{prop:connections} give us.

\subsection{Explicit solution for the quantum state, $\rho_i$}\label{app:QBER_state}

Recall that we proved above that $s_i = a + b + 2c \triangleq s, \, 2 \le i \le n+1$, by defining $x_i = a_i + b_i + c_i +d_i$, and $y_i = c_i + d_i$, and using Eqs.~\eqref{eq:aip1},~\eqref{eq:bip1},~\eqref{eq:cip1},~\eqref{eq:dip1}, to obtain $x_{i+1} + y_{i+1} = a+b+2c \equiv s$, for all $i \in \left\{2, 3, \ldots, n+1\right\}$, and that $s_1 = a_e + b_e + 2c_e$ for the elementary link. Let us now proceed to calculate the coefficients $a_i$, $b_i$, $c_i$ and $d_i$, all explicitly as a function of $i$, $1 \le i \le n+1$, and the system's loss and noise parameters.

\begin{proposition}\label{prop:aiplusbi}
$a_i + b_i \equiv z_i$ is given by,
\begin{equation}\label{eq:zi}
z_i = \nu\left(\frac{z_1}{\nu} \times \frac{s}{s_1}\right)^{2^{i-1}}, \, i \ge 2,
\end{equation}
where $z_1 = a_e + b_e$, $s_1 = a_e + b_e + 2c_e$, $\nu = s^2/(a+b)$, and $s \triangleq s_i$, for $i \ge 2$.
\end{proposition}
\begin{proof}
The proof follows by realizing that with the definitions in Eqs.~\eqref{eq:x} and~\eqref{eq:y}, $x_i - y_i = a_i + b_i$, and,
\begin{equation}
x_{i+1} - y_{i+1} = \frac{1}{s_i^2}(a+b)(x_i - y_i)^2.
\end{equation}
\end{proof}

\begin{remark}
Note that since $x_i + y_i = s_i$, and $x_i - y_i = z_i$, we have,
\begin{equation}
y_i = c_i + d_i = \frac{1}{2}\left[s_i - \nu\left(\frac{sz_1}{s_1\nu}\right)^{2^{i-1}}\right].
\end{equation}
As we will see in the next subsection, the error probability $Q_i$ depends only on $2c_i/s_i$---the fractional probability of the classical correlations when two copies of $\rho_{i-1}$ are connected. Note that $(a_i + b_i)$ is the sum fractional probability of the Bell states $|M^+\rangle$ ($a_i$) and $|M^-\rangle$ ($b_i$) when two copies of $\rho_{i-1}$ are connected, and $s_i = (a_i+b_i) + 2c_i$. Since we already have $c_i + d_i$ explicitly available, let us calculate $c_i - d_i \equiv u_i$.
\end{remark}

\begin{proposition}\label{prop:logistic}
The difference $c_i - d_i \equiv u_i$ can be found as the solution to the following quadratic difference equation,
\begin{equation}
w_{i+1} = w_r + 2(1 - 2w_r)w_i(1-w_i),\label{eq:diff}
\end{equation}
where $w_i \triangleq u_i/z_i$, $w_r = c/(a+b)$, and $w_1 = c_e/(a_e + b_e)$.
\end{proposition}

\begin{proof}
The proof follows from simply writing down $c_{i+1} - d_{i+1}$ using Eqs.~\eqref{eq:cip1} and~\eqref{eq:dip1}, substituting $w_i = u_i/z_i$, and simplifying.
\end{proof}

\begin{remark}\label{rem:logeq_specialcase}
The difference equation equation~\eqref{eq:diff} reduces to the famous Logistic Map, when $w_r = 0$. The solution to the logistic map $w_{i+1} = Rw_i(1-w_i)$, $w_i \in (0, 1)$, is in general chaotic, but for $R=2$ (which is exactly what ~\eqref{eq:diff} reduces to when $w_r = 0$) was found exactly by Ernst Schr{\" o}der in 1870, as:
\begin{equation}\label{eq:exactsolution}
w_i = \frac{1}{2}\left[1 - (1-2w_1)^{2^{i-1}}\right].
\end{equation}
\end{remark}

\begin{theorem}\label{thm:solutionlogisticmap}
The quadratic difference equation, $w_{i+1} = w_r + 2(1 - 2w_r)w_i(1-w_i)$, which is a variant of the logistic map $w_{i+1} = Rw_i(1-w_i)$ with $R=2$, can be exactly solved, and the solution is given by:
\begin{equation}
w_i = \frac{1}{2}\left[1 - \frac{1}{\beta}\left[\beta(1-2w_1)\right]^{2^{i-1}}\right],
\end{equation}
where $\beta = 1 - 2w_r$. This correctly reduces to~\eqref{eq:exactsolution} when $w_r = 0$.
\end{theorem}
\begin{proof}
See next Section for the proof.
\end{proof}

Next, we find $c_i$. We add the following two expressions:
\begin{eqnarray}
c_i + d_i &=& (s_i - z_i)/2, \,{\text{and}} \\
c_i - d_i &=& u_i = \frac{z_i}2\left[1 - \frac{1}{\beta}\left[\beta(1-2w_1)\right]^{2^{i-1}}\right],
\end{eqnarray}
and divide by $2$, to obtain:
\begin{equation}\label{eq:ci}
c_i = \frac{s_i}{4}\left[1 - \frac{z_i}{\beta s_i} \left[\beta(1-2w_1)\right]^{2^{i-1}}\right].
\end{equation}

At this point, since we have $c_i$, it is sufficient to calculate $Q_i$ (see next subsection). However, let us go ahead and evaluate $a_i$ and $b_i$ as well, so that we have a complete characterization of the quantum state $\rho_i$, which can be used to calculate other quantities of interest, such as the fidelity, entanglement of formation, etc.

Since we already have $a_i +b_i = z_i$ from Proposition~\ref{prop:aiplusbi}, we need to calculate $a_i - b_i$.
\begin{proposition}
$a_i - b_i \equiv v_i$ is given by the following recursion,
\begin{equation}
v_i = \frac{1}{s_i^2}(a-b)z_iv_i,
\end{equation}
which can be solved to obtain:
\begin{equation}
v_i = \left(\frac{a-b}{a+b}\right)^{i-1}\left(\frac{a_e-b_e}{a_e+b_e}\right)z_i,
\end{equation}
where $z_i$ is given by Eq.~\eqref{eq:zi}.
\end{proposition}
\begin{proof}
The proof follows simply by subtracting the expressions for $b_{i+1}$ from that of $a_{i+1}$, given in Proposition~\ref{prop:connections}, and simplifying.
\end{proof}

With that, we finally have the state $\rho_i$ as,
\begin{eqnarray}
\rho_i &=& \frac{1}{s_i}\left[a_i|M^+\rangle\langle M^+| + b_i|M^-\rangle\langle M^-| + c_i|\psi_0\rangle\langle\psi_0| \right.\nonumber\\
&&\left. +\; d_i|\psi_1\rangle\langle\psi_1| + d_i|\psi_2\rangle\langle\psi_2| + c_i|\psi_3\rangle\langle\psi_3|\right],\label{eq:stateexpression_app}
\end{eqnarray}
where $|\psi_0\rangle = |01,01\rangle$, $|\psi_1\rangle = |01,10\rangle$, $|\psi_2\rangle = |10,01\rangle$, $|\psi_3\rangle = |10,10\rangle$, $|M^{\pm}\rangle = \left[|\psi_2\rangle \pm |\psi_1\rangle\right]/\sqrt{2}$, $s_i = a_i + b_i + 2(c_i + d_i)$, and the coefficients given as:
\begin{eqnarray}
a_i &=& \frac{1}{2}\left[1 + \left(\frac{a-b}{a+b}\right)^{i-1}\left(\frac{a_e-b_e}{a_e+b_e}\right)\right]z_i, \nonumber\\
b_i &=& \frac{1}{2}\left[1 - \left(\frac{a-b}{a+b}\right)^{i-1}\left(\frac{a_e-b_e}{a_e+b_e}\right)\right]z_i, \nonumber\\
c_i &=& \frac{s_i}{4}\left[1 - \frac{z_i}{s_i(1-2w_r)} \left[(1-2w_r)(1-2w_1)\right]^{2^{i-1}}\right], \nonumber\\
d_i &=& \frac{s_i}{4} - \frac{z_i}{2}\left[1 - \frac{1}{2(1-2w_r)}\left[(1-2w_r)(1-2w_1)\right]^{2^{i-1}}\right],\nonumber
\end{eqnarray}
with $w_1 = c_e/(a_e+b_e)$, $w_r = c/(a+b)$, $s_1 = a_e+b_e+2c_e$, $s_i = s = a+b+2c$, $2 \le i \le n+1$, and $z_i$ given by,
\begin{equation}
z_i = \left(\frac{s^2}{a+b}\right)\left(\frac{1}{(1+2w_1)(1+2w_r)}\right)^{2^{i-1}}, \, i \ge 2,
\end{equation}
with $z_1 = a_e + b_e$. The expressions for $a_i$, $b_i$, $c_i$, and $d_i$ correctly reduce to $a_e$, $b_e$, $c_e$, and $0$, respectively, for $i=1$. As an example calculation, the fidelity of $\rho_i$ (with respect to $|M^+\rangle$), $F_i = \sqrt{\langle M^+ |\rho_i |M^+\rangle}$ is given by, $F_i = \sqrt{(a_i+d_i)/s_i}$.

\subsection{Evaluating the formula for QBER}

\begin{proposition}\label{prop:Qformula}
Assume that Alice and Bob have made a measurement on $\rho_i$, $i \in \left\{1, \ldots, n+1\right\}$. Conditioned on the fact that they get exactly one click each on their qubits (which happens with probability $P_1$, as proven in Proposition~\ref{prop:P1}), the probability $Q_i$, that they obtain a mismatched bit (a bit error) is given by,
\begin{equation}\label{eq:Q_exactform_app}
Q_i = \frac{1}{2}\left[1 - \frac{t_d}{t_r}\left(t_rt_e\right)^{2^{i-1}}\right], \, 1 \le i \le n+1,
\end{equation}
where $t_e = (a_e+b_e-2c_e)/(a_e+b_e+2c_e)$, $t_r = (a+b-2c)/(a+b+2c)$, and $t_d = ((q_1-q_2)/(q_1+q_2+q_3))^2$ are loss-noise parameters of detectors in the elementary links, memory nodes, and Alice-Bob, respectively.
\end{proposition}

\begin{proof}
The first step is to show that $Q_i$ can be expressed as follows:
\begin{equation}\label{eq:Q_zeta}
Q_i = \frac{1}{2}\left[1 - t_d(1 - 2\zeta_i)\right],
\end{equation}
where $\zeta_i = 2c_i/s_i$, $t_d = ((q_1-q_2)/(q_1+q_2+q_3))^2$. Since we have shown that $s_i=s=a+b+2c$, $i \ge 2$, and $s_1 = a_e+b_e+2c_e$, we only need to solve for $c_i$, in order to evaluate $Q_i$. In order to prove~\eqref{eq:Q_zeta}, we need to evaluate
\begin{eqnarray}
Q_i &=& \frac{1}{P_1} \biggl( {\rm Tr} \bigl[ \rho_i(M_{0101} + M_{1010} + \nonumber\\
&& \frac{1}{2}\left\{M_{1101}+M_{1110}+M_{0111}+M_{1011}+M_{1111}\right\} ) \bigr] \biggr) \nonumber
\end{eqnarray}
where the denominator $P_1 = {\rm Tr}[\rho_i(M_{0101} + M_{0110} + M_{1001} + M_{1010} + M_{1101} + M_{1110} + M_{0111} + M_{1011} + M_{1111})] = (q_1+q_2+q_3)^2$. We first note that $\rho_i$ is of the form,
\begin{eqnarray}
\rho_i &=& r_1|M^+\rangle\langle M^+| + r_2|M^-\rangle\langle M^-| + r_3|\psi_0\rangle\langle\psi_0| \nonumber\\
&& +\; r_4|\psi_1\rangle\langle\psi_1| + r_5|\psi_2\rangle\langle\psi_2| + r_6|\psi_3\rangle\langle\psi_3|,
\end{eqnarray}
with $\sum_{i=1}^6r_i = 1$. Noting that the relative contributions of $|\psi_0\rangle, |\psi_1\rangle, |\psi_2\rangle, |\psi_3\rangle$ in $\rho_i$ are $r_3, r_4 + (r_1+r_2)/2, r_5 + (r_1+r_2)/2$, and $r_6$ respectively, we now evaluate each of the $7$ terms in the expression for $Q_i$ as follows:
\begin{eqnarray}
{\rm Tr}(\rho_i M_{0101}) &=& q_1^2r_3 + q_1q_2\left[r_4 + \frac12(r_1+r_2) \right. \nonumber \\
&&\left. + r_5 + \frac12(r_1+r_2)\right] + q_2^2r_6, \nonumber\\
{\rm Tr}(\rho_i M_{1010}) &=& q_2^2r_3 + q_1q_2\left[r_4 + \frac12(r_1+r_2)\right. \nonumber\\
&&\left.+\; r_5 + \frac12(r_1+r_2)\right] + q_1^2r_6, \nonumber\\
\frac{1}{2}{\rm Tr}(\rho_i M_{1101}) &=& \frac{1}{2}\bigg[ r_3q_3q_1 + r_4q_3q_2 + r_5q_3q_1 + r_6q_3q_2 \nonumber \\
&& \left. + \left(\frac{r_1+r_2}{2}\right)q_3q_2 + \left(\frac{r_1+r_2}{2}\right)q_3q_1 \right], \nonumber \\
\frac{1}{2}{\rm Tr}(\rho_i M_{1110}) &=& \frac{1}{2}\bigg[ r_3q_3q_2 + r_4q_3q_1 + r_5q_3q_2 + r_6q_3q_1 \nonumber \\
&& \left. + \left(\frac{r_1+r_2}{2}\right)q_3q_2 + \left(\frac{r_1+r_2}{2}\right)q_3q_1 \right], \nonumber \\
\frac{1}{2}{\rm Tr}(\rho_i M_{0111}) &=& \frac{1}{2}\bigg[ r_3q_3q_1 + r_4q_3q_1 + r_5q_3q_2 + r_6q_3q_2 \nonumber \\
&& \left. + \left(\frac{r_1+r_2}{2}\right)q_3q_1 + \left(\frac{r_1+r_2}{2}\right)q_3q_2 \right], \nonumber \\
\frac{1}{2}{\rm Tr}(\rho_i M_{1011}) &=& \frac{1}{2}\bigg[ r_3q_3q_2 + r_4q_3q_2 + r_5q_3q_1 + r_6q_3q_1 \nonumber \\
&& \left. + \left(\frac{r_1+r_2}{2}\right)q_3q_2 + \left(\frac{r_1+r_2}{2}\right)q_3q_1 \right], \nonumber \\
\frac{1}{2}{\rm Tr}(\rho_i M_{1111}) &=& \frac{1}{2}\bigg[ r_3q_3^2 + r_4q_3^2 + r_5q_3^2 + r_6q_3^2 \nonumber \\
&& \left. + \left(\frac{r_1+r_2}{2}\right)q_3^2 + \left(\frac{r_1+r_2}{2}\right)q_3^2 \right]. \nonumber
\end{eqnarray}
Adding the above, and substituting $P_1 = (q_1+q_2+q_3)^2$, we get,
\begin{equation}
Q_i = \frac{(q_1-q_2)^2(r_3+r_6)+2q_1q_2 + (q_1+q_2)q_3 + \frac{q_3^2}{2}}{(q_1+q_2+q_3)^2}.
\end{equation}
Substituting $r_3 = r_6 = c_i/s_i$, defining $\zeta_i = 2c_i/s_i$, we get
\begin{eqnarray}
1-2Q_i &=& \frac{1}{(q_1+q_2+q_3)^2} \bigg[ (q_1+q_2+q_3)^2 - 2\zeta_i (q_1-q_2)^2 \nonumber \\
&& - 4q_1q_2 - 2(q_1+q_2)q_3 - q_3^2 \bigg] \nonumber \\
&=& \frac{1}{(q_1+q_2+q_3)^2} \bigg[ (q_1+q_2+q_3)^2 - 2\zeta_i (q_1-q_2)^2 \nonumber \\
&& - (q_1+q_2+q_3)^2 + (q_1 - q_2)^2 \bigg] \nonumber \\
&=& (1-2\zeta_i)\left(\frac{q_1-q_2}{q_1+q_2+q_3}\right)^2
\end{eqnarray}
Defining $t_d = ((q_1-q_2)/(q_1+q_2+q_3))^2$, Eq.~\ref{eq:Q_zeta} follows.

We now divide $2c_i$ (from Eq.~\eqref{eq:ci}) by $s_i$ to obtain,
\begin{equation}
\zeta_i = \frac{2c_i}{s_i} = \frac{1}{2}\left[1 - \frac{z_i}{\beta s_i} \left[\beta(1-2w_1)\right]^{2^{i-1}}\right].
\end{equation}
Substituting the expression for $z_i$ above, and realizing that $s_i = s$, $i \ge 2$, and $s_1 = a_e + b_e + 2c_e$, it is easy to obtain the expression for $Q_i$ in Eq.~\ref{eq:Q_exactform_app} after some algebraic manipulations. The $i=1$ case must be handled separately (since $s_1 \ne s_i, i \ge 2$), but the final expression in Eq.~\ref{eq:Q_exactform_app} is valid for all $i = 1, 2, \ldots, n+1$.

\end{proof}

The following corollary is an interesting consequence of Eq.~\ref{eq:Q_exactform_app}:
\begin{corollary}
The following law for {\em error propagation} holds through the successive connections of elementary links:
\begin{equation}\label{eq:errorpropagation_supp}
(1 - 2Q_{i+1}) = \frac{t_r}{t_d}(1-2Q_i)^2, 1 \le i \le n.
\end{equation}
An interesting thing to note about the error propagation is the constant $t_r = (1-2w_r)/(1+2w_r)$, which is a function of the parameter $2w_r = 2c/(a+b)$. We saw that when two pure bell states are `connected' by a linear-optic BSM with lossy-noisy detectors, $2c$ is the fractional probability that spills over into classical correlations (the nonentangled part), and $a+b$ is the fractional probability that goes into one of two entangled bell states.
\end{corollary}

Putting everything together, we finally have an expression for the secret-key rate,
\begin{equation}
R = \frac{P_1P_{\rm succ}R_2(Q_{n+1})}{2T_q}\;{\text{secret-key bits/s}},
\end{equation}
where $P_{\rm succ} = \left[4s\left(1-(1-4s_1)^M\right)\right]^{2^n}/4s$, $P_1 = (q_1+q_2)^2$, and $Q_{n+1} = \left[1 - \frac{t_d}{t_r}\left(t_rt_e\right)^{2^{n}}\right]/2$, are all defined in terms of the detector loss and noise parameters, and the total number of elementary links $N = 2^n$.

\section{Solution of the modified logistic map}\label{app:logisticmap}

In this section, we prove the following new variation of the logistic map, whose solutions are known to have chaotic behavior in general.
\begin{theorem}\label{thm:logisticmapproof}
The quadratic difference equation, $w_{i+1} = w_r + 2(1 - 2w_r)w_i(1-w_i)$, which is a variant of the logistic map $w_{i+1} = 2w_i(1-w_i)$ with $R=2$, can be exactly solved, and the solution is given by:
\begin{equation}
w_i = \frac{1}{2}\left[1 - \frac{1}{\mu}\left[\mu(1-2w_1)\right]^{2^{i-1}}\right], i \ge 1,
\end{equation}
where $\mu = 1 - 2w_r$, and the initial value $w_1$ specified.
\end{theorem}
\begin{proof}
We start with the solution to the standard logistic map with $R=2$, i.e., with $w_r=0$. The solution is given by:
\begin{equation}\label{eq:exactsolution}
w_i = \frac{1}{2}\left[1 - (1-2w_1)^{2^{i-1}}\right].
\end{equation}
We use the ansatz that the modified map has the solution of the form
\begin{equation}\label{eq:ansatz}
w_i = \frac{1}{2}\left[1 - (1-2w_1)^{2^{i-1+\xi_i}}\right].
\end{equation}
Inserting this into the difference equation, we get
\begin{eqnarray}
&&\frac{1}{2}\left[1 - (1-2w_1)^{2^{i+\xi_{i+1}}}\right] = \nonumber\\
&&w_r + \frac{(1 - 2w_r)}{2}\left[1 - (1-2w_1)^{2^{i+\xi_i}}\right]\,.
\end{eqnarray}
Letting $y_i=(1-2w_1)^{2^{i+\xi_i}}$ and $\mu=1-2w_r$, we obtain
\begin{equation}
y_{i+1}={\mu}^2y_i^2\,,
\end{equation}
which can be solved to obtain
\begin{equation}
y_i=\frac{1}{{\mu}^2}({\mu}^2y_1)^{2^{i-1}},\, i \ge 1,
\end{equation}
Using this to solve for $\xi_i$, we get
\begin{equation}
\xi_i = i - \log_2\left[ \frac{2^i\log_2({\mu}(1-2w_1))-\log_2({\mu}^2)}{\log_2(1-2w_1)}\right]\,.
\end{equation}
Finally, inserting the expression for $\xi_i$ into the ansatz, we obtain the following expression for $w_i$.
\begin{equation}
w_i=\frac{1}{2}\left[1-\frac{1}{{\mu}}({\mu}(1-2w_1))^{2^{i-1}}\right],\, i \ge 1.
\end{equation}

\end{proof}

\section{Derivation of the rate-loss envelope}\label{app:ratelossanalysis}

In subsection~\ref{app:ratelossenvelope1} of this Appendix, we will show that the key rate achieved over a range $L$, when divided up into $N$ equal segments, $R_N(L)$ can be upper bounded by a three-piece approximation $R_N^{(\rm UB)}(L)$. In subsection~\ref{app:ratelossenvelope2}, we will derive the envelope $R^{(\rm UB)}(L)$ of the three-piece upper bounds $R_N^{(\rm UB)}(L)$, which in turn is an upper bound to the true rate-loss envelope. Finally, in subsection~\ref{app:ratelossenvelope3}, we will derive an exact expression for the rate-loss envelope (assuming all detector dark clicks to be zero) and show that when an optimal number $N^*(L)$ of elementary links are employed at a given range $L$, the resulting rate-loss envelope $R^{(0)}(L) = A\eta^\xi$, where $\eta = e^{-\alpha L}$.

\subsection{Three-piece rate-loss upper bound for a given number of elementary links}\label{app:ratelossenvelope1}

In this section, we will first discuss the intuition behind why it is reasonable to expect that non-zero detector dark clicks cannot increase the secret-key rate achieved by the repeater protocol, i.e., $R_N(L) \le R_N^{(0)}(L)$. We will argue why a mathematically rigorous proof of above is not trivial, despite the fact that the statement sounds intuitively obvious. In the second part of this section, we will provide a proof of Theorem~\ref{thm:ratelossUB}, assuming $R_N(L) \le R_N^{(0)}(L)$ holds for all $N \ge 1$.

\subsubsection{Non-zero dark clicks can only decrease the secret-key rate: an intuitive argument}

Let us consider the model for a non-ideal single photon detector developed in Section~\ref{app:elemlink}. The ``no click'' and ``click'' events at the output of a single photon detector, of detection efficiency $\eta$ and dark click probability $P_d$, correspond to a two-outcome POVM $\{F_0,F_1\}$, with
\begin{eqnarray}
F_0 &=& (1-P_d)\Pi_0 + (1-A_d)\Pi_1 + (1-B_d)\Pi_2
\\
F_1 &=& P_d\Pi_0 + A_d\Pi_1 + B_d\Pi_2\;.
\end{eqnarray}
where,
\begin{eqnarray}
A_d &=& 1-(1-P_d)(1-\eta), \, {\text{and}}
\\
B_d &=& 1-(1-P_d)(1-\eta)^2\;.
\end{eqnarray}
In writing the above POVM elements, we have assumed that the quantum state $\rho$ impinging on the detector has no more than $2$ photons, which holds true for all the theoretical analysis in Section~\ref{sec:analysis} that assumed $p(2)=0$. Pictorially, this detection model is elucidated in Fig.~\ref{fig:detmodel}(a), where the lossy-noisy detector is modeled as outputting the Boolean OR of two binary-valued random variables $X$ and $Y$, where $X$ is the output of an ideal single photon detector ($\left\{|0\rangle\langle 0|, {\hat I}-|0\rangle\langle 0|\right\}$) preceded by a pure-loss beamsplitter of transmissivity $\eta$ upon which the input state $\rho$ is incident, and $Y$ is a binary-valued random variable that models dark clicks, is statistically independent of $X$, and satisfies ${\rm Pr}[Y=1] = P_d$. It is easy to see that this model is equivalent to the detection model shown in Fig.~\ref{fig:detmodel}(b), where a lossy-noiseless detector (detection efficiency $\eta$, zero dark-click probability) is followed by a binary-input binary-output discrete memoryless ``Z" channel.
\begin{figure}
\centering
\includegraphics[width=\columnwidth]{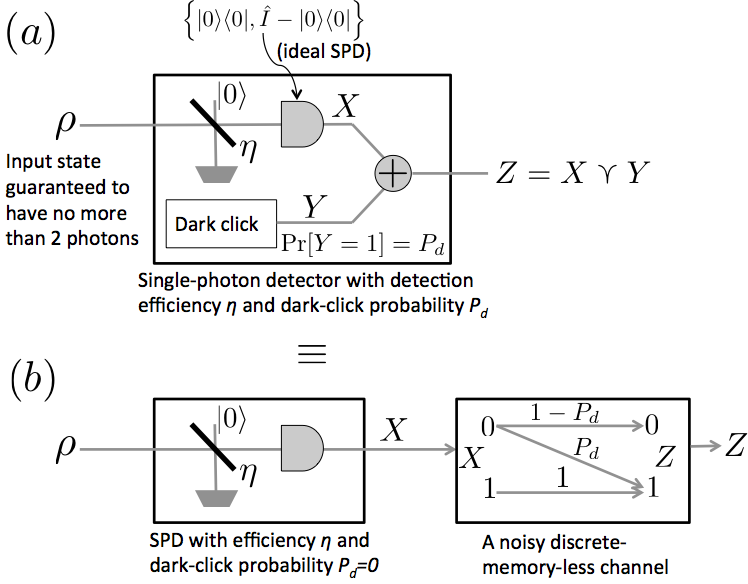}
\caption{Two equivalent models of a lossy-noisy single photon detector. $X, Y, Z \in \left\{0, 1\right\}$ are binary-valued random variables, and `$\curlyvee$' is the logical OR operation.}
\label{fig:detmodel}
\end{figure}

With the above two detection models applied to both single photon detectors of Alice, and both detectors of Bob, it is easy to see that a non-zero dark click probability at Alice's and Bob's detectors can be interpreted as a (random) local post processing of the raw classical data obtained by Alice and Bob when they (hypothetically) use zero-dark-click detectors. Since any local post-processing of their detection outcomes cannot increase the extractable secret-key rate, one concludes $R_N(L)$ is bounded above by the rate achieved with an $N$ link chain when Alice's and Bob's detectors have zero dark clicks. However, we need to prove $R_N(L) \le R_N^{(0)}(L)$, where$R_N^{(0)}(L)$ is the secret key rate when {\em all} the detectors in the system have zero dark click probability. So, we continue the argument above---that of using the equivalent interpretation of lossy-noise single photon detection depicted in Fig.~\ref{fig:detmodel}---for all the detectors used at the $N-1$ repeater nodes ($4(N-1)$ detectors) and at the centers of $N$ elementary links ($4NM$ single-frequency single-photon detectors, or $4N$ single-photon detectors that can spectrally resolve the $M$ orthogonal frequencies). Let us define $R_N^{(0), {\rm opt}}(L)$ to be the rate achievable when (a) {\em all} detectors in the system have zero dark clicks, and (b) optimal post-processing of all the detector outputs is used (note that Eve has access to most of these outputs as well except for those at Alice's and Bob's stations). Let us define $R_N^{{\rm opt}}(L)$ to be the rate achievable when (a) {\em all} detectors in the system have non-zero dark click probabilities ($P_e, P_r, P_d$, depending upon which detector), and (b) optimal post-processing of all the detector outputs is used. Note that not only Eve has access to most of these detector outputs (ones at repeater nodes and elementary link centers), she could in fact be using noiseless detectors and simulating dark clicks locally. Again, we can rigorously argue that:
\begin{equation}
R_N^{{\rm opt}}(L) \le R_N^{(0), {\rm opt}}(L),\label{eq:app1}
\end{equation}
since classical post-processing of the raw detector outputs (which affects only Alice's and Bob's raw classical data) cannot increase their extractable key rate. However, in our repeater protocol, we use a specific post-processing of the vector of detection outcomes at all the single photon detectors. Hence we have:
\begin{eqnarray}
R_N(L) &\le& R_N^{{\rm opt}}(L), \label{eq:app2}\,{\text{and}}\\
R_N^{(0)}(L) &\le& R_N^{(0), {\rm opt}}(L).\label{eq:app3}
\end{eqnarray}
Equations~\eqref{eq:app1},~\eqref{eq:app2} and~\eqref{eq:app3} are insufficient to conclude that $R_N(L) \le R_N^{(0)}(L)$.

\subsubsection{Proof of Theorem~\ref{thm:ratelossUB}}

In this section, we will prove that:
\begin{equation}
R_N^{(0)}(L) \le R_N^{({\rm UB})}(L) = \left\{
\begin{array}{ll}
R_{\rm max}, & {\text{for}}\; 0 \le L \le L^\prime, \\
\eta \left(AB^N\right), & {\text{for}}\; L^\prime < L < L_{\rm max}, \\
0, & {\text{for}}\; L \ge L_{\rm max},
\end{array}
\right.\label{eq:ratelossUB_app}
\end{equation}
with $L^\prime = -\log(\eta^\prime)/\alpha$, $\eta^\prime = (2/M\eta_e^2)^N$, and $R_{\rm max} = A\,(\eta_r^2 \lambda_m^2/2)^N$, where the constants $A$ and $B$ are given by, $A = \eta_d^2/(\eta_r^2 \lambda_m^2 T_q)$ and $B = \eta_r^2 \lambda_m^2 \eta_e^2 M/4$. Assuming that $R_N(L) \le R_N^{(0)}(L)$ holds $\forall N \ge 1$, the bound in Theorem~\ref{thm:ratelossUB} will follow.

The rate $R_N^{(0)}(L)$ assumes that $P_d = P_r = P_e = 0$, which implies $Q(N) = 0$, and hence $R_2(Q(N))=1$, $4s = \eta_r^2 {\lambda}_m^2 / 2$, and $4s_1 = \eta_e^2 \lambda^2 / 2 = \eta_e^2 \eta^{1/N}/2$, since $\lambda = \eta^{1/2N}$. Also, $P_1 = (q_1 + q_2)^2 = \eta_d^2$. Since $P_{s0} = 4s_1 < 1$, since it is a probability (of a BSM `success' on one of the frequency modes of one elementary link), with $M \ge 1$ and $N \ge 1$, we have that $\left(1-(1-4s_1)^M\right)^N \le 1$. Therefore,
\begin{eqnarray}
P_{\rm succ} &=& (4s)^{N-1} \left(1-(1-4s_1)^M\right)^N\\
&\le& (4s)^{N-1}.
\end{eqnarray}
It is now easy to derive a constant ($L$-independent) upper bound to $R_N^{(0)}(L)$, the first segment of $R_N^{({\rm UB})}(L)$.
\begin{eqnarray}
R_N^{(0)}(L) &=& \frac{P_1P_{\rm succ}R_2(Q(N))}{2T_q}\\
&=& \frac{\eta_d^2}{2T_q}P_{\rm succ} \\
&\le& \left(\frac{\eta_d^2}{\eta_r^2 \lambda_m^2 T_q}\right)\;\left(\frac{\eta_r^2 \lambda_m^2}{2}\right)^N \\
&=& A \left(\frac{\eta_r^2 \lambda_m^2}{2}\right)^N \equiv R_{\rm max},
\end{eqnarray}
where $A = \eta_d^2/(\eta_r^2 \lambda_m^2 T_q)$. Next, we observe that $(1-4s_1)^M \ge 1-4Ms_1$ for $M \ge 1$. In other words, $1 - (1-4s_1)^M \le 4Ms_1$. Hence, we have
\begin{eqnarray}
P_{\rm succ} &=& (4s)^{N-1} \left(1-(1-4s_1)^M\right)^N\\
&\le& (4s)^{N-1}\; \left(4Ms_1\right)^N \\
&=& (4s)^{N-1}\; \left(\frac{M\eta_e^2 \eta^{1/N}}{2}\right)^N\\
&=& (4s)^{N-1}\; \left(\frac{M\eta_e^2}{2}\right)^N \eta \\
&=& \eta \left(\frac{1}{4s}\right) \left(4s\frac{M\eta_e^2}{2}\right)^N\\
&=& \eta \left(\frac{2}{\eta_r^2 \lambda_m^2}\right) \left(\frac{M\eta_e^2\eta_r^2\lambda_m^2}{4}\right)^N.
\end{eqnarray}
Therefore, we have,
\begin{eqnarray}
R_N^{(0)}(L) &=& \frac{P_1P_{\rm succ}R_2(Q(N))}{2T_q}\\
&=& \frac{\eta_d^2}{2T_q}P_{\rm succ} \\
&\le& \eta \left(A B^N\right),
\end{eqnarray}
where $A = \eta_d^2/(\eta_r^2 \lambda_m^2 T_q)$, and $B = \eta_r^2 \lambda_m^2 \eta_e^2 M/4$, which gives us the linear rate-transmittance (second segment) of the upper bound $R_N^{({\rm UB})}(L)$. The third segment of $R_N^{({\rm UB})}(L)$ is trivial since $R_N(L) = 0$ for $L \ge L_{\max}$.

\subsection{Envelope of the three-piece rate-loss upper bounds}\label{app:ratelossenvelope2}

In this section, we will prove Theorem~\ref{thm:ratelossUB_exp}, i.e., derive the envelope of $R_N^{({\rm UB})}(L)$ over all $N \ge 1$. The main step will be to prove (see below) that the locus of the corner points $\left\{X_N\right\}$ is given by $A\eta^t$ with $t = {\log\left(\eta_r^2 \lambda_m^2/2\right)}/{\log\left(2/M\eta_e^2\right)} \le 1$. Next we argue that since the line segments connecting $X_N$ and $Y_N$ are proportional to $\eta$ (i.e., $\eta \left(AB^N\right)$), that the locus of the corner points $\left\{Y_N\right\}$ cannot be above the locus of the corner points $\left\{X_N\right\}$ (since $t \le 1$). We thereby conclude that the envelope of the functions $R_N^{({\rm UB})}(L)$ over all $N \ge 1$, is given by $A\eta^t$. Finally, since $R_N(L) \le R_N^{(0)}(L) \le R_N^{(\rm UB)}(L)$, given $R(L)$ is the envelope of $R_N(L)$ over all $N \ge 1$ and given $R^{(\rm UB)}(L)$ is the envelope of $R_N^{(\rm UB)}(L)$ over all $N \ge 1$, we get the statement of Theorem~\ref{thm:ratelossUB_exp}, i.e., $R(L) \le R^{(\rm UB)}(L) = A\eta^t$.

Let us now prove the only step we left open above, that the locus of the corner points $\left\{X_N\right\}$ is given by $A\eta^t$ with $t = {\log\left(\eta_r^2 \lambda_m^2/2\right)}/{\log\left(2/M\eta_e^2\right)}$. The proof follows simply by calculating the coordinates of $X_N(\eta^\prime, R^\prime)$, where $\eta^\prime$ is given by equating the first two segments of $R_N^{(\rm UB)}(L)$, and solving for $\eta$:
\begin{equation}
(AB^N)\eta^\prime = A\left(\frac{\eta_r^2\lambda_m^2}{2}\right)^N,
\end{equation}
which yields $\eta^\prime = (\frac{2}{M\eta_e^2})^N$. Clearly, $R^\prime = R_{\rm max} = A({\eta_r^2\lambda_m^2}/2)^N$. Eliminating $N$ from the expressions of $\eta^\prime(N)$ and $R^\prime(N)$ by taking logarithms and dividing, it is simple to obtain the solution of the locus of the points $\left\{X_N\right\}$ as $R^\prime = A(\eta^\prime)^t$, where $A = \eta_d^2/(\eta_r^2 \lambda_m^2 T_q)$, and $t = {\log\left(\eta_r^2 \lambda_m^2/2\right)}/{\log\left(2/M\eta_e^2\right)}$. Hence proved.

\subsection{Exact expression for the rate-loss envelope}\label{app:ratelossenvelope3}

In this section, we will prove Theorem~\ref{thm:xiformula}, i.e., derive $R^{(0)}(L) = A\eta^\xi$, the exact solution of the envelope of $R_N^{(0)}(L)$ over all $N \ge 1$, where $A = \eta_d^2/(\eta_r^2 \lambda_m^2 T_q)$, and the exponent $\xi$ is given by:
\begin{equation}
\xi = \frac{\log\left[\beta \left(1 - (1-\gamma z)^M\right)\right]}{\log z},
\label{eq:xi_app}
\end{equation}
where $z$ is the unique solution of the following transcendental equation in the interval $(0, 1)$:
\begin{eqnarray}
&&\left(1-(1-\gamma z)^M\right)\log \left[\beta(1 - (1-\gamma z)^M)\right] \nonumber\\
&&= \gamma Mz\log z \left(1-\gamma z\right)^{M-1},
\end{eqnarray}
with, $\beta = \eta_r^2\lambda_m^2/2$, and $\gamma = \eta_e^2/2$.

We can express $R_N^{(0)}(L) \equiv y = P_1P_{\rm succ}/2T_q = \eta_d^2 P_{\rm succ}/2T_q$ as:
\begin{equation}
y = A\left[\beta\left(1 - \left(1 - \gamma x^{1/N}\right)^M\right)\right]^N,
\end{equation}
where $x = \eta$ is the channel transmittance, $A = \frac{\eta_d^2}{\eta_r^2\lambda_m^2T_q}$, $\beta = \eta_r^2\lambda_m^2/2$, and $\gamma = \eta_e^2/2$. Substituting $t = 1/N$, the envelope of $R_N^{(0)}(L)$ over $N \ge 1$ is given by the simultaneous solution of $f(x,y,t) = 0$ and $\partial f(x,y,t) / \partial t = 0$, where
\begin{equation}
f(x,y,t) = \left(\frac{y}{A}\right)^t - \beta\left(1 - (1-\gamma x^t)^M\right),
\end{equation}
with $t \equiv 1/N \in (0,1]$. The two simultaneous equations are thus given by:
\begin{eqnarray}
z^t &=& \beta\left(1 - (1-\gamma x^t)^M\right),\,{\text{and}}\label{eq:z1}\\
z^t \log z &=& \beta\gamma Mx^t \log x \left(1-\gamma x^t\right)^{M-1},\label{eq:z2}
\end{eqnarray}
where $z \equiv y/A$. We will next argue that the unique solution to Eqs.~\eqref{eq:z1} and~\eqref{eq:z2} must be of the form, $z = x^\xi$. To do so, let us differentiate $z$ with respect to $x$ in Eq.~\eqref{eq:z1}, which yields
\begin{equation}
z^{t-1}\frac{dz}{dx} = \beta\gamma M\left(1-\gamma x^t\right)^{M-1}x^{t-1}.
\end{equation}
Substituting $\beta\gamma Mx^t \left(1-\gamma x^t\right)^{M-1} = z^t \log z/ \log x$ from Eq.~\eqref{eq:z2}, we get
\begin{equation}
\frac{dz}{z\log z} = \frac{dx}{x\log x},
\end{equation}
taking an indefinite integral of which yields:
\begin{equation}
\log \log z - \log \log z_0 = \log \log x - \log \log x_0,
\end{equation}
where $z_0$ and $x_0$ are constants to be determined, by substituting the solution back into $f(x,y,t) = 0$. Simplifying the above, we obtain,
\begin{equation}
\log\left(\frac{\log z}{\log x}\right) = \log\left(\frac{\log z_0}{\log x_0}\right),
\end{equation}
or $z = x^\xi$, with $\xi = {\log z_0}/{\log x_0}$. Finally, we substitute $z = x^\xi$ into Eq.~\eqref{eq:z1} and solve to obtain the expression for $\xi$ as shown in Eq.~\eqref{eq:xi_app}, and hence obtaining $y = Ax^\xi$. Hence, we have $R^{(0)}(L) = A\eta^\xi$, the exact solution of the envelope of $R_N^{(0)}(L)$ over all $N \ge 1$.
\end{document}